\documentclass[11pt]{scrartcl}
\usepackage[totalwidth=460pt,totalheight=680pt]{geometry}
\usepackage[all=normal,paragraphs=tight,bibliography=tight]{savetrees}
\usepackage{microtype}
\usepackage[utf8]{inputenc}
\usepackage[T1]{fontenc}
\usepackage[english]{babel}
\usepackage{csquotes}

\usepackage[cmex10]{amsmath}
\usepackage{amsthm}
\usepackage{amssymb}
\usepackage{comment}
\usepackage{url}
\usepackage{xspace}
\usepackage{enumitem}

\usepackage{tikz}
\usetikzlibrary{calc,arrows}
\usepackage{graphicx}
\usepackage{float}

\usepackage[backend=bibtex,style=numeric,bibstyle=numeric,firstinits=true,doi=false,isbn=false,url=false,maxbibnames=99]{biblatex}
\renewbibmacro{in:}{}
\AtEveryBibitem{\clearname{editor}}

\newtheorem{theorem}{Theorem}
\newtheorem{lemma}[theorem]{Lemma}

\newtheorem{corollary}[theorem]{Corollary}
\newtheorem{fact}[theorem]{Fact}

\theoremstyle{definition}
\newtheorem{definition}[theorem]{Definition}

\newtheorem{conjecture}{Conjecture}

\newcommand{\iflncs}[2]{\ifthenelse{\equal{\uselncs}{yes}}{#1}{#2}}

\newcommand{\uselncs}{no}

\newcommand{\tw}{\ensuremath{\mathtt{tw}}\xspace}
\newcommand{\pw}{\ensuremath{\mathtt{pw}}\xspace}
\newcommand{\td}{\ensuremath{\mathtt{td}}\xspace}
\newcommand{\poly}{\mathrm{poly}}
\newcommand{\Oh}{\ensuremath{\mathcal{O}}}

\newcommand{\oeff}{\ensuremath{o_{\text{eff}}}}
\newcommand{\N}{\ensuremath{\mathbb{N}}}
\newcommand{\Z}{\ensuremath{\mathbb{Z}}}

\newcommand{\tree}{\texttt{tree}}
\newcommand{\chld}{\texttt{chld}}
\newcommand{\tail}{\texttt{tail}}
\newcommand{\roots}{\texttt{roots}}
\newcommand{\Aa}{\mathbf{A}}
\newcommand{\Ff}{\mathbf{F}}
\newcommand{\Tt}{\mathbf{T}}

\newcommand{\cfont}[1]{\textnormal{\textbf{#1}}}
\newcommand{\mystrut}{\vphantom{tg}} 
\newcommand{\optional}[1]{\textcolor[rgb]{0.3,0.3,0.3}{#1}}
\newcommand{\lred}{\leq_L}

\newcommand{\exampleProblem}{\textsc{3Coloring}\xspace}
\newcommand{\tdExample}{\textnormal{td-}\exampleProblem}
\newcommand{\tdEx}[1]{\textnormal{td-}\exampleProblem{}[\ensuremath{#1}]}
\newcommand{\pwExample}{\textnormal{pw-}\exampleProblem}
\newcommand{\pwEx}[1]{\textnormal{pw-}\exampleProblem{}[\ensuremath{#1}]}
\newcommand{\twExample}{\textnormal{tw-}\exampleProblem}
\newcommand{\twEx}[1]{\textnormal{tw-}\exampleProblem{}[\ensuremath{#1}]}

\newcommand{\NAuxSAplain}[3]{\cfont{NAuxSA[}#1,#2,#3\cfont{]}}
\newcommand{\NAuxSA}[3]{\cfont{NAuxSA[}\underset{time}{#1\mystrut},\underset{space}{#2\mystrut},\underset{height}{#3\mystrut}\cfont{]}}
\newcommand{\NAuxPDA}[2]{\cfont{NAuxPDA[}\underset{time\mystrut}{#1},\underset{space\mystrut}{#2}\cfont{]}}

\newcommand{\DTiSp}[2]{\cfont{D[}\underset{time}{#1\mystrut},\underset{space}{#2\mystrut}\cfont{]}}
\newcommand{\DTime}[1]{\cfont{D[}\underset{time}{#1\mystrut}\cfont{]}}
\newcommand{\DSpace}[1]{\cfont{D[}\underset{space}{#1\mystrut}\cfont{]}}
\newcommand{\NTiSp}[2]{\cfont{N[}\underset{time}{#1\mystrut},\underset{space}{#2\mystrut}\cfont{]}}

\newcommand{\NSpace}[1]{\cfont{N[}\underset{space}{#1\mystrut}\cfont{]}}

\newcommand{\ATiSz}[2]{\cfont{A[}\underset{time}{#1\mystrut},\underset{treesize}{#2\mystrut}\cfont{]}}
\newcommand{\ASpSz}[2]{\cfont{A[}\underset{space}{#1\mystrut},\underset{treesize}{#2\mystrut}\cfont{]}}
\newcommand{\paraDTiSp}[2]{\cfont{D[}#1,#2\cfont{]}}
\newcommand{\paraNTiSp}[2]{\cfont{N[}#1,#2\cfont{]}}
\newcommand{\NSC}[1]{\cfont{NSC}^{#1}}
\newcommand{\SACq}[1]{\cfont{SAC}_{quasi}^{#1}}

\newcommand{\NL}{\cfont{NL}\xspace}
\newcommand{\SC}{\cfont{SC}\xspace}

\newcommand{\app}{$\spadesuit$}

\newcommand{\defparproblemu}[4]{
  \vspace{2mm}
\noindent\fbox{
  \begin{minipage}{0.96\textwidth}
  \begin{tabular*}{\textwidth}{@{\extracolsep{\fill}}lr} #1 & {\textbf{Parameter:}} #3 \\ \end{tabular*}
  {\textbf{Input:}} #2  \\
  {\textbf{Question:}} #4
  \end{minipage}
  }
  \vspace{2mm}
}

\title{On space efficiency of algorithms working on structural decompositions of graphs
\thanks{
	{A 
	preliminary version of this paper appeared in the Proceedings of the 33$^{\textrm{rd}}$  Symposium on		
	Theoretical Aspects of Computer Science (STACS’16).}\newline
  The research of both authors is supported by Polish National Science Centre grant DEC-2013/11/D/ST6/03073.
  During the work on these results, Micha\l{} Pilipczuk held a post-doc position at Warsaw Center of Mathematics and Computer Science, and was supported by the Foundation for Polish Science via the START stipend programme.}}
\author{
  Micha\l{} Pilipczuk\thanks{Institute of Informatics, University of Warsaw, Poland, \texttt{\{michal.pilipczuk,m.wrochna\}@mimuw.edu.pl}.}
  \and 
  Marcin Wrochna$^{\dagger}$
  }

\date{}

\begin{document}

\maketitle

\begin{abstract}
Dynamic programming on path and tree decompositions of graphs is a technique that is ubiquitous in the field of parameterized and exponential-time algorithms. However, one of its drawbacks is that the space usage is exponential in the decomposition's width. Following the work of Allender et al. [Theory of Computing,~'14], we investigate whether this space complexity explosion is unavoidable. Using the idea of reparameterization of Cai and Juedes~[J. Comput. Syst. Sci.,~'03], we prove that the question is closely related to a conjecture that the {\textsc{Longest Common Subsequence}} problem parameterized by the number of input strings does not admit an algorithm that simultaneously uses $\cfont{XP}$ time and $\cfont{FPT}$ space. Moreover, we complete the complexity landscape sketched for pathwidth and treewidth by Allender et al. by considering the parameter {\em{tree-depth}}. We prove that computations on tree-depth decompositions correspond to a model of non-deterministic machines that work in polynomial time and logarithmic space, with access to an auxiliary stack of maximum height equal to the decomposition's depth.
Together with the results of Allender et al., this describes a hierarchy of complexity classes for polynomial-time non-deterministic machines with different restrictions on the access to working space, which mirrors the classic relations between treewidth, pathwidth, and tree-depth.

\end{abstract}

\newpage
\section{Introduction}\label{sec_intro}
Treewidth is a graph parameter that measures how much a graph resembles a tree. Intuitively, a graph has treewidth $s$ if it can be decomposed into parts (called {\em{bags}}) of size at most $s+1$ so that this decomposition, called a {\em{tree decomposition}}, has the shape of a tree. While initially introduced by Robertson and Seymour in their Graph Minors project~\cite{RobertsonS86}, treewidth has found numerous applications in the field of algorithms. This is mostly due to the fact that many problems that are intractable on general graphs, become efficiently solvable on graphs of small treewidth using the principle of dynamic programming. Theorems of Courcelle~\cite{Courcelle90} and of Arnborg et al.~\cite{ArnborgLS91} explain that every decision or optimization problem expressible in Monadic Second Order logic can be solved in time $f(s)\cdot n$ on graphs of treewidth $s$ and size $n$, for some function $f$. While $f$ can be non-elementary, for many classic problems, like {\textsc{Vertex Cover}}, {\textsc{3Coloring}}, or {\textsc{Dominating Set}}, the natural dynamic programming approach yields an algorithm with running time $\Oh(c^s\cdot n)$ for some typically small constant $c$. This is important from the point of view of applications, because dynamic programming procedures working on tree decompositions often serve as critical subroutines in more complex algorithms, such as subexponential algorithms derived using the technique of {\em{bidimensionality}}~\cite{DemaineFHT05}, or approximation schemes that use Baker's approach~\cite{Baker94}. Algorithms working on tree decompositions are usually analyzed in the paradigm of {\em{parameterized complexity}}, where the width of the decomposition is the considered parameter. Nowadays, the study of problems parameterized by structural measures of the input, such as treewidth, is a classic topic of study in this field. We refer to textbooks~\cite{platypus,DowneyF13,grohe:book} for a broad introduction, and to a recent survey of Langer et al.~\cite{LangerRRS14} for more specific results.

A certain limitation of dynamic programming on a tree decomposition is that it uses space exponential in its width. In practical applications this is often a prohibitive factor, because the machine is likely to simply run out of space much before the elapsed time exceeds tolerable limits. Therefore, recently there is much focus on reducing the space complexity of exponential-time algorithms to polynomial, even at the cost of slightly worsening the time complexity~\cite{AustrinKKM13,narrow-sieves,FominKLPS15,FurerY14,LokshtanovN10,Nederlof13}. Here, the usage of algebraic tools proved to be an extremely useful approach. Unfortunately, algorithms working on tree decompositions of graphs remain a family where virtually no progress has been achieved in this matter. Therefore, a natural question arises: Can we reduce the space complexity of algorithms working on tree decompositions while keeping (or moderately worsening) their time complexity? The first explicit statement of this question known to us is due to Lokshtanov et al.~\cite{LokshtanovMS11}, who sketched how using a simple tradeoff one can obtain polynomial space complexity while increasing the time complexity to $2^{\Oh(s^2)}+\Oh(n^2)$. Later, the same question was asked by Langer et al.~\cite{LangerRRS14}.

Following early completeness results of Monien and Sudborough~\cite{MonienS85} on bandwidth-constrained problems and of Gottlob et al.~\cite{GottlobLS01} on conjunctive queries of bounded treewidth, Allender et al.~\cite{AllenderCLPT14} recently initiated the systematic study of satisfaction complexity in small, but non-constant path- and treewidth. Essentially, they observe that CSP-like problems---say, {\textsc{3Coloring}} for concreteness\footnote{Allender et al. use {\textsc{SAT}} parameterized by treewidth/pathwidth of its primal graph as an exemplary problem, but {\textsc{SAT}} and {\textsc{3Coloring}} can be easily seen to be equivalent under logspace reductions; see Lemma~\ref{lem:equivalenceOfProblems}. In this paper, we prefer to use {\textsc{3Coloring}} as an exemplary hard CSP-like problem.}---when limited to instances of small treewidth or pathwidth, are complete for certain complexity classes under logspace reductions. More precisely, when the input graph is equipped with a path decomposition of width at most $s(n)\geq \log n$, for some fixed constructible function $s$ of the input size, then {\textsc{3Coloring}} (denoted in this case as \pwExample{}[$s$]) is complete for the class $\NTiSp{\poly}{s(\poly)}$: problems that admit non-deterministic algorithms working simultaneously in polynomial time and space $\Oh(s(\poly(n))$. Similarly \twExample{}[$s$], where $s(n)$ bounds the width of a given tree decomposition, is complete for the class $\NAuxPDA{\poly}{s(\poly)}$; the difference with $\NTiSp{\poly}{s(\poly)}$ is that the algorithm can use an auxiliary push-down of unlimited size, to which read/write access is only from the top. We remark that Allender et al. focus on a different characterization by \emph{semi-unbounded fan-in} (SAC) circuits; this characterization will not be relevant to our work, and hence we omit it.
While they state their results only for $s(n)=\log^k n$, the proof works in the more general setting given below.

\begin{theorem}[\cite{AllenderCLPT14}]\label{thm:pwCompleteness}
Let $s(n) \geq \log n$ be a nice function\footnote{By a {\em{nice}} function we mean a function $s$ that is constructible and such that $s(n)/\lg n$ is non-decreasing.}.\\
Then \pwExample{}$[s]$ is complete for $\NTiSp{\poly}{s(\poly)}$ under logspace reductions, whereas \twExample{}$[s]$ is complete for $\NAuxPDA{\poly}{s(\poly)}$ under logspace reductions.
\end{theorem}
Thus, the feasibility of various space-time tradeoffs for algorithms working on tree/path decompositions is equivalent to inclusions of corresponding complexity classes. For instance (assuming for conciseness $s(n^c)=\Oh(s(n))$ for each constant~$c$, e.g., $s(n)=\log^{k} n$ for some $k$), \pwExample{}$[s]$ is solvable in:
\begin{itemize}
	\item time $2^{o(s(n) \log n)}$ and space $2^{o(s(n))}$ if and only if
	$\NTiSp{\poly}{s} \subseteq \DTiSp{2^{ o(s \cdot \log)} }{ 2^{o(s)} }$;
	\item time $2^{\Oh(s(n))}$ and space $\poly(n)$ if and only if
	$\NTiSp{\poly}{s} \subseteq \DTiSp{2^{\Oh(s)}}{\poly}$.
\end{itemize}
Similar statements can be inferred for treewidth. In contrast, the best known determinization results come from a brute-force approach or from Savitch's theorem%
\footnote{There is also a simple simulation approach that tries each possible sequence of non-deterministic choices, proving $\NTiSp{t}{s} \subseteq \DSpace{t}$, but in our context it is always worse than Savitch's theorem.}~\cite{savitch1970relationships}, yielding respectively (for constructible $t(n)\geq n$ and $s(n)\geq \log t(n)$):
\begin{eqnarray*}
&\NTiSp{t}{s} \subseteq \DTime{2^{\Oh(s)}} = \DTiSp{2^{\Oh(s)}}{\optional{2^{\Oh(s)}}}; \\
&\NTiSp{t}{s} \subseteq \DSpace{s\cdot \log t} = \DTiSp{\optional{2^{\Oh(s\cdot \log t)}}}{s\cdot \log t}.
\end{eqnarray*}
In this manner, Allender et al. conclude that, intuitively speaking, achieving better time-space tradeoffs for algorithms working on path and tree decompositions of small width would require developing a general technique of improving upon the tradeoff of Savitch. As Lipton~\cite{lipton2010savitch} phrased it, \enquote{one of the biggest
embarrassments of complexity theory is the fact that Savitch's theorem has
not been improved \textelp{}. Nor has anyone proved that it is tight}. 

Allender et al. argue that such an improvement would contradict certain rescaled variants of known conjectures about the containment of time- and space-constrained classes, in particular the assumption that $\cfont{NL}\nsubseteq \cfont{SC}$; more precisely that the rescaled hierarchies built on top of these two classes are somehow orthogonal---we refer to~\cite{AllenderCLPT14} for details. We consider the study of Allender et al. not as a definite answer in the topic, but rather as an invitation to a further investigation of links between the introduced conjectures and more well-established complexity assumptions.

\paragraph*{Our Contribution.}
In the {\textsc{Longest Common Subsequence}} problem (LCS), we are given an alphabet $\Sigma$ and $k$ strings $s_1,s_2,\ldots,s_k$ over $\Sigma$, and the task is to compute the longest sequence of symbols of $\Sigma$ that appears as a subsequence in each $s_i$. The applicability of the LCS problem in, e.g., computational biology, motivated many to search for faster or more space-efficient algorithms, as the classical dynamic programming solution, running in time and space $\Oh(n^k)$ (where $n$ is the length of each string and $k$ the number of strings) is often far from practical. From the point of view of parameterized complexity, LCS parameterized by $k$ is W[$t$]-hard for every level $t$~\cite{BodlaenderDFW95}, remains W[1]-hard even for a fixed-sized alphabet~\cite{Pietrzak03}, and is W[1]-complete when parameterized jointly by $k$ and $\ell$, the target length of the subsequence~\cite{Guillemot11}. In a recent breakthrough, Abboud et al.~\cite{AbboudBW15} proved that the existence of an algorithm with running time $\Oh(n^{k-\varepsilon})$, for any $\varepsilon>0$, would contradict the Strong Exponential Time Hypothesis. As far as the space complexity is concerned, only modest progress has been achieved: The best known result is an algorithm of Barsky et al.~\cite{BarskySTU07}, which improves the space complexity to $\Oh(n^{k-1})$. This motivates us to formulate the following conjecture.

\begin{conjecture}\label{conj:lcs}
There is no algorithm for LCS that works in time $n^{f(k)}$ and space $f(k) \poly(n)$ for a computable function $f$, where $k$ is the number of input strings and $n$ the total input length.
\end{conjecture}

Quite surprisingly, we show that Conjecture~\ref{conj:lcs} is closely related to the question about time-space tradeoffs for algorithms working on path decompositions of bounded width.
\begin{theorem}\label{thm:main}
The following statements hold:
\begin{itemize}
\item Conjecture~\ref{conj:lcs} implies the following: There is no algorithm for \pwExample working in time $2^{\Oh(s)} \poly(n)$ and space $2^{\oeff(s)} \poly(n)$ (for all values of width $s$).
\item Conjecture~\ref{conj:lcs} is implied by the following: There is an unbounded, computable function $g$, for which \pwExample cannot be solved in time $2^{s\cdot g(s)} \poly(n)$ and space $2^{\oeff(s)} \poly(n)$ (for all values of width $s$).
\end{itemize}
\end{theorem}
Thus, Conjecture~\ref{conj:lcs} is sandwiched between a weaker statement that it is impossible to achieve subexponential space complexity while keeping single exponential time complexity, and a stronger statement that this holds even if we allow the time complexity exponent to increase by some (arbitrarily slowly growing) computable function of the width.

To prove Theorem~\ref{thm:main}, we use the results of Elberfeld et al.~\cite{ElberfeldST15} who showed a completeness result for LCS. Viewed from this perspective, Conjecture~\ref{conj:lcs} is equivalent to a statement in parameterized complexity about the impossibility of determinization results improving upon Savitch's theorem (we defer the definitions and discussion to Section~\ref{sec:parameterized}).
Using the ideas of Cai and Juedes~\cite{CaiJ03} that connected subexponential complexity to fixed-parameter tractability, we consider a reparameterized version of \pwExample.
This allows us to compare questions concerning time-space tradeoffs for \pwExample and determinization of $\NTiSp{t}{s}$ classes to those concerning parameterized classes and the complexity of LCS.
In particular, we show that Conjecture~\ref{conj:lcs} implies $\NL\not\subseteq\DTiSp{\poly}{\poly\log}$ (the latter class being usually called $\SC$)
and is implied by a rescaled version of the following stronger variant: $\NL\not\subseteq\DTiSp{2^{o(\log^2 n)}}{n^{o(1)}}$.

\bigskip

In the second part of this work, we complement the findings of Allender et al.~\cite{AllenderCLPT14} by considering the graph parameter {\em{tree-depth}}. For a graph $G$, its tree-depth is equal to the minimum height of a rooted forest $\mathcal{T}$ whose ancestor-descendant closure contains $G$ as a subgraph; the forest $\mathcal{T}$ is then called a {\em{tree-depth decomposition}} of $G$. Tree-depth of a graph is lower bounded by its pathwidth and upper bounded by its treewidth times $\lg n$. Our motivation for considering this parameter is two-fold. First, recent advances have uncovered a wide range of topics where tree-depth appears naturally. For instance, tree-depth plays an important role in the theory of sparse graphs developed by Ne\v{s}et\v{r}il and Ossona de Mendez~\cite{sparsity}, is the key factor in classification of homomorphism problems that can be solved in logspace~\cite{ChenM14}, and corresponds to classes of graphs where the expressive power of First-Order logic and Monadic Second-Order logic coincides~\cite{ElberfeldGT12}. It was also rediscovered several times in different contexts and under different names: {\em{minimum elimination tree height}}~\cite{pothen}, {\em{ordered chromatic number}}~\cite{KatchalskiMS95}, {\em{vertex ranking}}~\cite{BodlaenderDJKKMT98}, or the maximum number of introduce nodes on a root-to-leaf path in a tree decomposition of a graph~\cite{FurerY14}.

Second, algorithms working on tree-depth decompositions of small height model generic exponential-time Divide\&Conquer algorithms. In this approach, after finding a small, usually balanced separator $S$ in the graph, the algorithm tries all possible ways a solution can interact with $S$, and solves connected components of $G-S$ recursively. A run of such an algorithm naturally gives rise to a tree-depth decomposition of the graph, where $S$ is placed on top of the decomposition, and decompositions of the components of $G-S$ are attached below it as subtrees. The maximum total number of separator vertices handled at any moment in the recursion corresponds to the height of the decomposition. Thus, many classic Divide\&Conquer algorithms, including the ones derived for planar graphs using the Lipton-Tarjan separator theorem~\cite{LiptonT80}, can be reinterpreted as first building a tree-depth decomposition of the graph using a separator theorem, and then running the algorithm on this decomposition.

Most importantly for us, when implemented using recursion, the algorithms working on tree-depth decompositions run in polynomial space. For instance, such an algorithm for {\textsc{3Coloring}} on a tree-depth decomposition of depth $s$ runs in time $3^s\cdot \poly(n)$ and space $\Oh(s+\log n)$ (see Lemma~\ref{lem:deter-3col}), which places \tdExample{}$[s]$ in $\DTiSp{2^{\Oh(s)}\poly}{s+\log}=\DSpace{s+\log n}$. This is immediate for CSP-like problems like {\textsc{3Coloring}}, but recently F\"urer and Yu~\cite{FurerY14} showed that algebraic transforms can be used to reduce the space usage to polynomial in $n$ also for other problems, like counting perfect matchings or dominating sets. In Section~\ref{sec:domset}, we describe how this approach gives an $3^s\cdot \poly(n)$-time $\poly(n)$-space algorithm for {\textsc{Dominating Set}} in more detail, and then improve the space usage to $\Oh(s \cdot \log n)$ using Fourier transforms and the Chinese remainder theorem. This means that the reduction of space complexity that is conjectured to be impossible for treewidth and pathwidth, actually is possible for tree-depth. Therefore, we believe that it is useful to study the model of computations standing behind low tree-depth decompositions, in order to understand how it differs from the models for treewidth and pathwidth. 

Consequently, mirroring Theorem~\ref{thm:pwCompleteness}, we prove that computations on tree-depth decompositions exactly correspond to the class $\NAuxSA{\poly}{\log}{s}$: problems which can decided by a non-deterministic Turing Machine that uses polynomial time and logarithmic working space, but also has access to an auxiliary stack of maximum height $s$. The stack can be freely read by the machine, just like the input tape, but the write access to it is only via \texttt{push} and \texttt{pop} operations.
\begin{theorem}\label{thm:main-td}
Let $s(n) \geq \log^2 n$ be a nice function.
Then \tdExample{}$[s]$ is complete for\linebreak $\NAuxSA{\poly}{\log}{s(\poly)}$ under logspace reductions.
\end{theorem}
Thus, computations on tree-depth and path decompositions differ by the access restrictions to $\Oh(s)$ space used by the machine. While for pathwidth this space can be accessed freely, for tree-depth all of the space apart from an $\Oh(\log n)$ working buffer has to be organized in a stack.

The proof of Theorem~\ref{thm:main-td} largely follows the approach of Akatov and Gottlob~\cite{AkatovG10}, who proved a different completeness result for the class $\NAuxSA{\poly}{\log}{\log^2}$, which they call $\cfont{DC}^1$.
The main idea is to regularize the run of the machine so that the push-pop tree has the rigid shape of a full binary tree. Then we can use this concrete structure to ``wrap around'' gadgets encoding an accepting run of a regularized NAuxSA machine.
However, the motivation in the work of Akatov and Gottlob was answering conjunctive queries in a hypergraph by exploiting a kind of balanced decomposition, and hence the problem proven to be complete for $\cfont{DC}^1$ is a quite general and expressive problem originating in database motivations; see~\cite{akatov2010exploiting,AkatovG10} for details. 
In our setting, in order to get a reduction to \exampleProblem, we need to work more to encode an accepting run. In particular, to encode each part of the computation where no push or pop is performed, instead of producing a single atom in a conjunctive query, we use computation gadgets that originate in Cook's proof of the NP-completeness of SAT.
The assumption that the computation has a polynomial number of steps is essential here for bounding the tree-depth of each such gadget.
This way, Theorem~\ref{thm:main-td} presents a more natural complete problem for $\cfont{DC}^1$.

Another difference is that Theorem~\ref{thm:main-td} works for any well-behaved function $s(n)\geq \log^2 n$, as opposed to the bound $s(n)=\log^2 n$ inherent to the problem considered by Akatov and Gottlob. 
For this, the crucial new idea is to increase the working space of the machine to $s(n)/\log n$ in order to be able to perform regularization -- a move that looks dangerous at first glance, but turns out not to increase the expressive power of the computation model. This proves the following interesting by-product of our work.

\begin{theorem}\label{thm:space-increase}
Let $s(n)\geq \log^2 n$ be a nice function.
Then $$\NAuxSA{\poly}{\log}{s(\poly)} = \NAuxSA{\poly}{s(\poly)/\log}{s(\poly)}.$$
\end{theorem}

The following determinization result for NAuxSA machines follows from the observation that \tdExample{}$[s]$ can be solved in $\DTiSp{2^{\Oh(s)}\poly}{s+\log}=\DSpace{s+\log n}$.

\begin{theorem}\label{thm:determinization}
Let $s(n)\geq \log^2(n)$ be a nice function. Then $$\NAuxSA{\poly}{\log}{s(\poly)}\subseteq \DSpace{s(\poly)}.$$
\end{theorem}

Theorem~\ref{thm:determinization} for $s(n)=\log^2 n$ also follows from the work of Akatov and Gottlob~\cite{AkatovG10}. Observe that now the justification for the assumption $s(n)\geq \log^2 n$ becomes apparent: for, say, $s(n)=\log n$, the theorem would state that $\cfont{L}=\cfont{NL}$, a highly unexpected outcome.

We find Theorem~\ref{thm:determinization} interesting, because a naive simulation of the whole configuration space for NAuxSA would require space exponential in $s$. It appears, however, that the exponential blow-up of the space complexity can be avoided. We do not see any significantly simpler way to prove this result other than going through the \tdEx{s} problem, and hence it seems that the tree-depth view gives a valuable insight into the computation model of NAuxSA.

The classic relations between treewidth, pathwidth and tree-depth are, through completeness results, mirrored in a hierarchy between NAuxPDA, N, and NAuxSA classes, as detailed in the concluding section.
In particular, this answers a question of Akatov and Gottlob~\cite{akatov2010exploiting,AkatovG10} about the relation of $\NAuxSA{\poly}{\log}{\poly \log}$ to other classes in $\cfont{NP}$.

Finally, using Theorem~\ref{thm:main-td} we also give an alternative view on NAuxSA computations using alternating Turing machines in Theorem~\ref{thm:alternation}, answering another question of Akatov and Gottlob.
From this point of view, Theorem~\ref{thm:determinization} is immediate.

\paragraph*{Outline.} In Section~\ref{sec:prelims} we give preliminaries, in Section~\ref{sec:parameterized} we prove Theorem~\ref{thm:main} and related results, 
and in Section~\ref{sec:treedepth} we prove Theorem~\ref{thm:main-td} and derive corollaries from it (Theorems~\ref{thm:space-increase} and~\ref{thm:determinization}).
Section~\ref{sec:domset} describes the algorithm for \textsc{Dominating Set} in graphs of low tree-depth. 
We finish by concluding remarks in Section~\ref{sec:conc}.
The discussion of how standard NP-hardness reductions preserve parameters linearly is deferred to Appendix~\ref{app:preserve}.

\section{Preliminaries}\label{sec:prelims}
\subsection{Reductions and complexity classes}
For two languages $P,Q$, we write $P \lred Q$ when $P$ is logspace reducible to $Q$.
Most of the complexity classes we consider are closed under logspace reductions.

Because we handle various measures of complexity and compare a wide array of classes that bound two measures simultaneously, we introduce the following notation.
A complexity class is first described by the machine model: \cfont{D}, \cfont{N}, \cfont{A} denote deterministic, non-deterministic, and alternating (see~\cite{Ruzzo80}) Turing machines, respectively.
Then, in square brackets, bounds on complexity measures are described (up to constant factors) as a list of functions with the name of the measure it bounds underneath.
All functions except the symbol $f$ (which we reserve for classes in parameterized complexity) are functions of the input size $n$.
For example, for $t,s:\N\to \N$, $\DSpace{s}$ denotes the class of languages recognizable by deterministic Turing machines using at most $\Oh(s(n))$ space, usually known as DSpace$(s(n))$;
similarly $\NTiSp{t}{s}$ is often denoted NTiSp$(t(n),s(n))$.
We write $\lg$ for the logarithm with base 2, $\log(n)$ when the base is irrelevant and $\poly(n)$ for $n^{\Oh(1)}$.
As customary for the $\Oh$-notation, a complexity class stated with a bound that is a family of functions (instead of a single function) is defined as the sum of classes over all functions in the family.
For example, $\DTime{\poly} = \bigcup_{k\in\N} \DTime{n^k} = \cfont{P}$ and $\NSpace{\log}=\NL$.

An auxiliary push-down or stack is denoted as AuxPDA or AuxSA, respectively: the difference is that a push-down can only be read at the top, while a stack can be read just as a tape (both can be written to only by pushing and popping symbols at the top), see e.g.~\cite{VinayC90}.
The measure named \emph{height} is the maximum height of the push-down or stack.

We say a function $s:\mathbb{N} \to \mathbb{N}$ is \emph{constructible} if there is a Turing machine which given a number $n$ in unary outputs $s(n)$ in unary using logarithmic space; in particular, this implies $s(n)\leq \poly(n)$.
A function $s$ is \emph{nice} if it is constructible and $\frac{s(n)}{\lg n}$ is non-decreasing; note that this implies that $s(n)$ itself is also non-decreasing.
For simplicity, we will assume all functions $s: \N \to \N$ describing complexity bounds to be nice.

Note that logspace reductions can blow-up instance sizes polynomially, hence the closure of $\NTiSp{\poly}{s}$ under such reductions is $\NTiSp{\poly}{s(\poly)}$, for example.
These are equal for functions $s(n)$ such that $s(\poly(n)) \leq \Oh(s(n))$ (that is, if for every $c>0$ there is a $d>0$ such that $s(n^c)\leq d s(n)$).
This includes $\lg^k(n)$ for any $k\geq 1$ and $\lg n \lg\lg n$, for example; however, one can construct artificial examples that are polylogarithmically bounded but fail to have this property.

\subsection{Structural parameters}
We recall the definitions of treewidth, pathwidth and tree-depth.
For conciseness, we will refer to the certifying structures defined below as \emph{decompositions} for all three parameters.

\newcommand{\T}{\mathcal{T}}

\begin{definition}[treewidth]
A \emph{tree decomposition} of an undirected graph $G$ is a tree $\T$ together with a collection of sets of vertices of $G$ (called \emph{bags}) $X_t$ indexed by nodes $t\in\T$, such that:
\begin{itemize}[itemsep=0pt]
\item every vertex of $G$ is in at least one bag;
\item for every edge $uv$ of $G$, there is a bag containing both $u$ and $v$; and
\item for every vertex $v$ of $G$, the set $\{t\in\T \mid v \in X_t\}$ induces a connected subtree of $\T$.
\end{itemize}
The width of a tree decomposition is defined as $\max_{t\in\T} |X_t| -1$.
The \emph{treewidth} of $G$ is the minimum width over all possible tree decompositions of $G$.
\end{definition}

\begin{definition}[pathwidth]
A \emph{path decomposition} of an undirected graph $G$ is a tree decomposition $(\T,(X_t)_{t\in\T})$ in which $\T$ is a path.
The \emph{pathwidth} of $G$ is the minimum width over all possible path decompositions of $G$.
\end{definition}

\begin{definition}[tree-depth]
A \emph{tree-depth decomposition} of an undirected graph $G$ is a rooted forest $\T$ (a disjoint union of rooted trees) together with a bijective mapping $\mu$ from the vertices of $G$ to the nodes of $\T$, such that for every edge $uv$ of $G$, $\mu(u)$ is an ancestor of $\mu(v)$ or $\mu(v)$ is an ancestor of $\mu(u)$ in $\T$.
The \emph{depth} of a rooted forest is the largest number of nodes on a path between a root and a leaf.
The \emph{tree-depth} of $G$ is the minimum depth over all possible tree-depth decompositions of $G$.
\end{definition}

For technical reasons, we assume that in all given tree and path decompositions, $|\T| \leq 2|V(G)|^2$;
standard methods allow to prune any decomposition to this size in logspace, see~\cite[Lemma 13.1.2]{Kloks94}.

\medskip

For a graph problem, such as \exampleProblem, a structural parameter $\pi\in\{\td,\pw,\tw\}$, and a nice function $s:\N\to\N$, we define $\pi$-\exampleProblem{}$[s]$ to be the decision problem where given an instance $G$ of \exampleProblem and a $\pi$-decomposition of $G$, we ask whether the decomposition has width at most $s(|V(G)|)$ and $G$ is a yes-instance of \exampleProblem.
The assumption that a decomposition is given on input is to factor away the complexity of finding it, which is a problem not directly relevant to our work. Note that the validity and width/depth of a decomposition given in any natural encoding can easily be checked in logarithmic space.

Observe also that for any $c>0$, $\pi$-\exampleProblem{}$[s(n)]$ is equivalent to $\pi$-\exampleProblem{}$[s(n^c)]$ under logspace reductions. Namely, a reduction from $\pi$-\exampleProblem{}$[s(n)]$ to $\pi$-\exampleProblem{}$[s(n^c)]$ is trivial, while the reverse reduction follows easily by padding: adding isolated vertices up to size $n^c$ that do not change the answer nor the value of $\pi$.
Since we assume $s$ to be nice, we have $\frac{s(n)}{\lg n} \leq \frac{s(n^c)}{\lg n^c}$, hence $c\cdot s(n) \leq s(n^c)$ for any $c\geq 1$. This implies that $\pi$-\exampleProblem{}$[c\cdot s(n)]$ is equivalent to $\pi$-\exampleProblem{}$[s(n)]$.

\medskip

A hierarchy between these classes immediately follows from well-known inequalities between the parameters $\td$, $\pw$, $\tw$, shown in the next lemma and corollary.
Every tree-depth decomposition can be turned into a path decomposition by taking bags to be the vertex sets of all the root-to-leaf paths, and ordering them as in a left-to-right scan of the tree.
Every path decomposition is trivially a tree decomposition.
Every tree decomposition allows to find small separators, which can be used to recursively build a tree-depth decomposition, yielding the last inequality---details can be found e.g. in~\cite[Corollary 2.5]{NesetrilM06}, we show how to execute them effectively below for completeness.

\begin{lemma}\label{lem:td_pw_tw}
	There is a constant $c\in\mathbb{N}$ such that for any graph $G$,
	$\td(G) \geq \pw(G) \geq \tw(G) \geq \td(G) /(c\cdot \log |V(G)|)$.
	Furthermore, each inequality is certified by an algorithm that transforms the respective graph decompositions in logspace.
\end{lemma}

\begin{corollary}\label{cor:td_pw_tw}
Let $s:\N\to\N$ be a nice function. Then
$$ \tdEx{s} \lred \pwEx{s} \lred \twEx{s} \lred \tdEx{s\cdot \log}.$$
\end{corollary}
\begin{proof}[Proof of Lemma~\ref{lem:td_pw_tw}]
	The algorithms for the first two inequalities are trivial.
	For the third inequality, a straightforward implementation would be problematic because of recursion and the need to remember a subset of vertices. We now show how to circumvent these issues.
	
	Let $\left(\T,(X_t)_{t\in\T}\right)$ be a given tree decomposition of a graph $G$ of width $k$.
	Elberfeld et al.~\cite[Theorem 14]{ElberfeldJT12} showed that there is a constant $c\in\mathbb{N}$ and a logspace algorithm that given a tree $\T$, outputs a width-3 tree decomposition $\left(\mathcal{S},(Y_s)_{s\in\mathcal{S}}\right)$ of $\T$ such that $\mathcal{S}$ is a full binary tree of depth $c \cdot \log |V(\T)|$ (their implementation in fact uses a circuit model even more restrictive than logspace).
	Let $Z_s=\bigcup_{t\in Y_s} X_t$; then it is easy to check that $\left(\mathcal{S}, (Z_s)_{s\in\mathcal{S}}\right)$ is a tree decomposition of $G$ of width at most $4k+3$, also computable in logspace. Note that $\mathcal{S}$ is rooted, so we can consider the ancestor relation on it.
	
	Since $\mathcal{S}$ has logarithmic depth, we can construct the following tree-depth decomposition of $G$.
	For $s\in V(\mathcal{S})$, let $\widetilde{Z}_s$ be the set of those vertices of $Z_s$, for which $s$ is the top-most node of $\mathcal{S}$ to whose bag $s$ belongs. Observe that $\{\widetilde{Z}_s\colon s\in V(\mathcal{S})\}$ is a partition of $V(G)$.
	Let $\mathcal{S}'$ be the tree obtained by replacing every node $s$ of $\mathcal{S}$ by a path $P_s$ of $|\widetilde{Z}_{s}|$ nodes, respecting the ancestor relation (so that the last vertex of $P_s$ becomes the parent of the first vertex of $P_{s'}$ for every child $s'$ of $s$).
	Consider any mapping $\mu: V(G) \to V(\mathcal{S'})$ which bijectively assigns vertices in each $\widetilde{Z}_{s}$ to nodes of $P_s$ in an arbitrary order. We claim that $(\mathcal{S},\mu)$ defines a tree-depth decomposition.
	To see this, consider any edge $uv$ of $G$.
	The vertices $u,v$ must be contained in some common bag $X_t$ and hence in some bag $Z_s$, $s\in V(\mathcal{S})$.
	Let $Z_{s(u)},Z_{s(v)}$ be the topmost bags containing $u,v$ respectively, then $u\in \widetilde{Z}_{s(u)}, v\in \widetilde{Z}_{s(v)}$.
	Both $s(u)$ and $s(v)$ must be ancestors of $s$ in $\mathcal{S}$, and hence they are themselves related by the ancestor relation.
	Since the ancestor relation was preserved by the construction, $\mu(u)$ is related to $\mu(v)$.
	This shows correctness.
	
	Membership in $\widetilde{Z}_s$, as well as $|\widetilde{Z}_s|$ can be calculated on the fly in logspace, hence it is straightforward to perform the whole construction in logspace.
	The depth of $\mathcal{S'}$ is at most $3k \cdot c \cdot \log |\T|$.
	Since we assumed that in all given decompositions $|\T|\leq 2 |V(G)|^2$, the depth is $\Oh(k \cdot \log |V(G)|)$.
\end{proof}

\subsection{Equivalence of problems}
We say that a reduction between two graph problems \emph{preserves structural parameters} (linearly) if for each parameter $\pi\in\{\tw,\pw,\td\}$ there is an integer $c\in\N$ such that for any instance with graph $G$, the graph $H$ produced by the reduction satisfies $\pi(H) \leq c \cdot \pi (G)$, and moreover a decomposition of $G$ of width/depth at most $s$ can be transformed in logspace into a decomposition of $H$ of width/depth at most $c\cdot s$.
Many known NP-hardness reductions can be shown to have this property, in particular those that replace each vertex or edge with a gadget of bounded size (see the descriptions of `local replacement' and `component design' methods in the classical work of Garey and Johnson~\cite{GareyJ79}).
For example, \exampleProblem and variants of SAT are equivalent in all our theorems, while \textsc{Vertex Cover} or \textsc{Dominating Set} (defined in~\cite{GareyJ79}) are at least as hard.
The proofs are deferred to the appendix.

\begin{definition}
Let $\phi$ be a CNF formula.
The \emph{primal (Gaifman) graph of $\phi$} is the graph with a vertex for each variable of $\phi$ and an edge between every pair of variables that appear together in some clause.
The \emph{incidence graph of $\phi$} is the bipartite graph with a vertex for each clause and each variable of $\phi$ and an edge between each clause and every variable contained in that clause.
\end{definition}

\begin{lemma}[\app]\label{lem:equivalenceOfProblems}
The following problems are equivalent under logspace reductions that preserve structural parameters: \textsc{3Coloring}, CNF-SAT (using a decomposition of the primal graph), $k$-SAT (using a decomposition of either the primal or incidence graph) for each $k\geq 3$.

Furthermore, the following problems admit logspace reductions that preserve structural parameters from the above problems: \textsc{Vertex Cover}, \textsc{Independent Set}, \textsc{Dominating Set}.
\end{lemma}

We will often consider problems like $\pi$-3-SAT or $\pi$-CNF-SAT, for $\pi\in\{\tw,\pw,\td\}$, in which case we always mean the width/depth of a given decomposition of the primal graph of the formula.

\subsection{Cook's theorem with bounded space}
In our reductions we will need to describe Turing machine computations using CNF formulas, just as in Cook's theorem on the NP-completeness of CNF-SAT.
It has already been observed by Monien and Sudborough~\cite{MonienS85} that Cook's reduction applied to machines with bounded space yields formulas of bounded width.
The only difference is that in this setting a machine's worktape space bound can be significantly shorter than the input word---the read-only tape on which the input is placed must be treated differently. One standard solution would be to modify the machine to make it oblivious, i.e., simulate tape operations so that head movements are independent of the input;
the reduction can then encode appropriate input symbols exactly where they would be read.
We employ a different approach by encoding the reading process directly into the formula, providing the input to each computation step with a copy of the following simple \emph{random access gadget}.
This has the advantage of making our reductions slightly more explicit and adaptable.

\begin{lemma}[Random access gadget]\label{lem:ramGadget}
For every $n\in\mathbb{N}$, there is a 3-CNF formula including named variables: $x_0,\dots,x_{n-1}$ (`{input}'), $y_0,\dots,y_{\lceil \lg n \rceil-1}$ (`{index}'), and $z$ (`{output}'), such that:
every assignment of the named variables can be extended to a satisfying assignment if and only if it satisfies $z=x_{\bar{y}}$.
Here $\bar{y}$ is the number encoded in binary by the index variables (we require $z=0$ if $\bar{y}\geq n$).
The formula has $\Oh(n)$ variables in total, tree-depth $\Oh(\log n)$ (of the primal graph), and can be constructed using $\Oh(\log n)$ space, given $n$.
\end{lemma}
\begin{proof}
Construct a full binary tree of variables, of depth $\lceil \lg n \rceil$.
Let the root be $z$ and name the leaves $x_0,x_1,\dots, x_{2^{\lceil \lg n \rceil}-1}$.
Add clauses of size $1$ requiring $x_i=0$ for $i\geq n$.
Introduce new variables $y_0,\dots,y_{\lceil \lg n \rceil-1}$.
For each internal variable $v$ at level $j$ of the tree, with children $v_0,v_1$, enforce that $v=v_{0}$ if $y_j$ is false and $v=v_{1}$ otherwise (using four clauses of size $3$ on the variables $\{v,v_0,v_1,y_j\}$).
This enforces that any assignment to the input and index variables can be extended to a satisfying assignment in exactly one way, in which furthermore $z=x_{\bar{y}}$.
A tree-depth decomposition of logarithmic depth is obtained from the full binary tree simply by mapping the index variables to a path of length $\lceil \lg n \rceil$ attached above the root of the tree.
\end{proof}

The following lemma shows more precisely (for our needs) how a Turing machine computation can be encoded in a formula (a \emph{computation gadget}) -- the crucial part of Cook's theorem.
In reductions involving stack machines it will also describe fragments of computation without any push/pop operation. The contents of the stack will be considered as a separate read-only input tape, which we treat differently because, while smaller in size, the content is not given to the reduction.

Note that Lemmas~\ref{lem:compGadget} and~\ref{lem:equivalenceOfProblems}, applied for $h=0$, immediately give the first part of Theorem~\ref{thm:pwCompleteness}: for nice $s(n)\geq \log n$, \pwExample{}$[s]$ is complete for $\NTiSp{\poly}{s(\poly)}$ under logspace reductions.

\begin{lemma}[Computation gadget]\label{lem:compGadget}
Let $M$ be a non-deterministic Turing machine over alphabet $\Sigma$ with two read-only input tapes and one work tape.
Given an input word $\alpha$ over $\Sigma$ of length $n$ and integers $s,t,h$ such that $\lg n+\lg h\leq \Oh(s)$, one can using $\Oh(\log(n+s+t+h))$ space output a CNF formula such that:
\begin{itemize}
\item The formula has $\Oh(t\cdot(s+h+n))$ variables, including named variables $u_1,\dots,u_{s'}$, $v_1,\dots,v_{s'}$, $w_1,\dots,w_{h'}$, where $s'\in \Theta(s)$ and $h'=h\cdot \lceil \lg |\Sigma|\rceil$. These variables respectively describe two configurations $\mathbf{u}$, $\mathbf{v}$ of $M$ (up to $s$ symbols of the working tape, heads' positions encoded in binary, and the state), and a word $\bar{w}$ over $\Sigma$ of length $h$.
\item Any assignment to the named variables can be extended to a satisfying assignment iff the computation of $M$ on inputs $\alpha$ and $\bar{w}$ can lead (by some sequence of non-deterministic choices) from the configuration $\mathbf{u}$ to the configuration $\mathbf{v}$, using at most $t$ steps and $s$ space.
\item The formula's primal graph has pathwidth $\Oh(s+h)$ and tree-depth  $\Oh(s\cdot\log(n+s+t+h)+h)$. Moreover, appropriate decompositions can be output within the same space bound.
\end{itemize}
\end{lemma}
\begin{proof}
	For simplicity of presentation, we assume input tapes use the binary alphabet. A larger alphabet can be reduced by encoding each symbol using a block of $\lceil \lg |\Sigma|\rceil$ bits, and applying straightforward modifications to the machine $M$.
	
	As in Cook's original proof, we create $t$ blocks of $s'\in\Theta(s)$ variables each, describing the configuration at each step.
	The first and last blocks contain variables $(u_i)_{1\leq i\leq s'}$ and $(v_i)_{1\leq i\leq s'}$, respectively, to encode configurations $\mathbf{u}$ and $\mathbf{v}$.
	Additionally, variables $(w_i)_{1\leq i\leq h}$ are created.
	
	We may assume that machine $M$ always keeps track of the indices on which the heads of the input tapes are placed. These indices are encoded in binary in pre-defined buffers on the working tape ({\em{index buffers}}), and each time the head on an input tape is moved, the machine updates the index. To encode reading access to the input tapes in the formula, for each step and each input tape, we create a copy of the random access gadget of Lemma~\ref{lem:ramGadget}. For the first input tape, the gadget has its input variables fixed with the bits of the word $\alpha$ (given to the reduction). For the second input tape, the gadget has its input variables identified with $w_1,\dots,w_{h}$. The index variables are identified with the variables of the block that encode the contents of respective index buffers. The machine behavior is then encoded with clauses binding variables of two consecutive blocks, including the output bit $z$ of each random-access gadget, exactly as in Cook's proof: the clauses verify the correctness of the transition. In doing this, we allow at each step a transition that does not change the configuration in any way. Transitions that would move the working tape's head outside the first $s$ symbols are not allowed. This construction makes the formula satisfiable exactly with assignments in which the $t$-th block describes a configuration reachable in at most $t$ steps and $s$ space (from the configuration described by $u$ variables). 
	
	It remains to bound the pathwidth and the tree-depth of the constructed formula's primal graph.
	To construct a path decomposition, create $t-1$ bags $A_1,A_2,\dots,A_{t-1}$, where each $A_i$ contains a pair of consecutive configuration blocks $i$ and $i+1$ (that is, variables describing the configuration just before and after a single computation step) and all the variables $(w_i)_{1\leq i\leq h}$.
	Consider one of the two random-access gadgets created for step $i$, and let $B_1,\dots,B_b$ be the bags of the provided path decomposition of this gadget ($|B_j|= \Oh(\log n+\log h)\leq \Oh(s)$ for each $1\leq j\leq b$). Similarly, let $B'_1,\dots,B'_{b'}$ be the bags of the provided path decomposition of the second gadget ($|B'_j|\leq \Oh(s)$). 
	Then the final decomposition is obtained by adding bags $A_i \cup B_1,\dots, A_i \cup B_b,A_i \cup B'_1, \dots, A_i \cup B'_{b'}$ immediately after $A_i$.
	From the construction it follows that each bag contains at most $\Oh(s'+h+s)=\Oh(s+h)$ variables and each clause binds a set of variables contained in one of the bags.
	For the tree-depth bound, consider the above decomposition with the variables $w_1,\dots,w_h$ removed.
	Its width is $\Oh(s)$ and hence using Lemma~\ref{lem:td_pw_tw}, we can obtain a tree-depth decomposition of the formula's primal graph of depth $\Oh(s\cdot \log(n+s+t+h))$ if we removed the $w$ variables.
	The variables can then be reintroduced by placing them atop all others in the decomposition, in any order, which increases the tree-depth by at most $h$.
\end{proof}

\section{Connections with tradeoffs for LCS}\label{sec:parameterized}

\newcommand{\rfont}[1]{\textcolor{blue}{#1}}
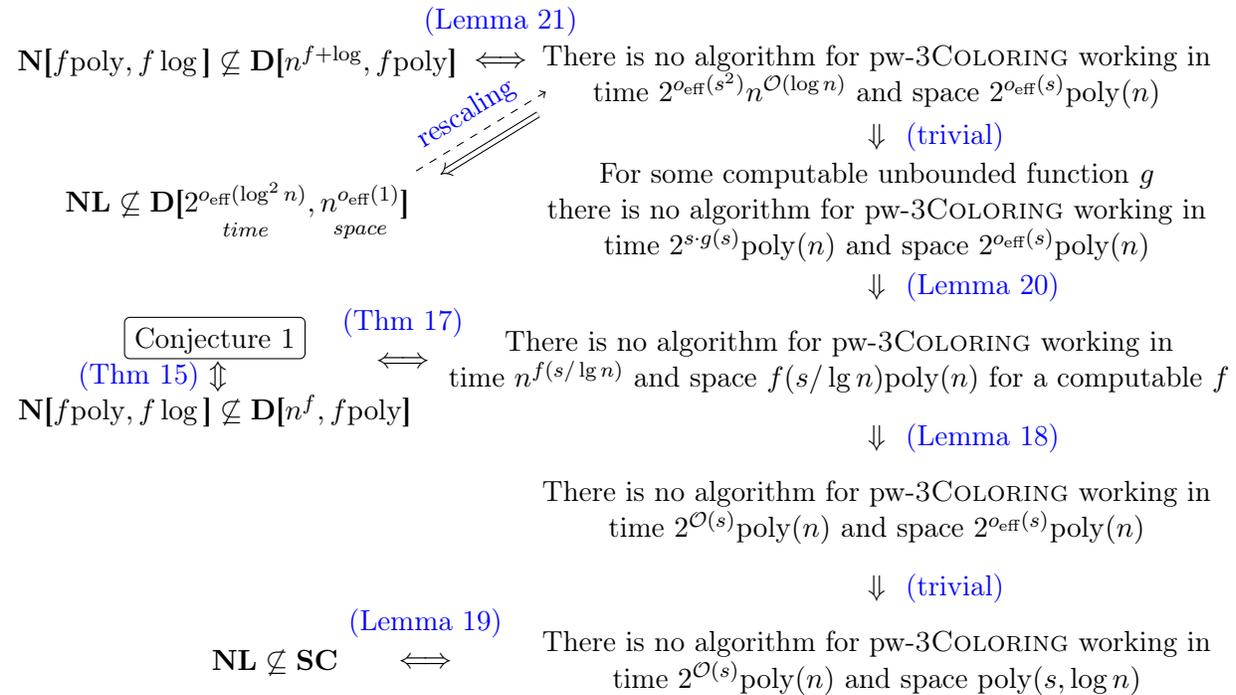
\begin{figure}[h]
	\begin{tikzpicture}
	\tikzset{every node/.style={align=center}}
	
	\node at (0,-0.2) {
		There is no algorithm for \pwExample working in\\
		time $2^{\oeff(s^2)}n^{\Oh(\log n)}$ and space $2^{\oeff(s)}\poly(n)$};
	\node[label=right:\rfont{(trivial)}] at (0,-1) {$\Downarrow$};
	\node at (0,-2) {
		For some computable unbounded function $g$\\
		there is no algorithm for \pwExample working in\\
		time $2^{s\cdot g(s)}\poly(n)$ and space $2^{\oeff(s)}\poly(n)$}; 
	\node[label=right:\rfont{(Lemma~\ref{para-strong-sandwich})}] at (0,-3) {$\Downarrow$};
	\node at (-0.5,-4) {
		There is no algorithm for \pwExample working in\\
		time $n^{f(s/\lg n)}$ and space $f(s/\lg n)\poly(n)$ for a computable $f$};
	\node[label=right:\rfont{(Lemma~\ref{para-weak-sandwich})}] at (0,-5) {$\Downarrow$}; 
	\node at (0,-6) {
		There is no algorithm for \pwExample working in\\
		time $2^{\Oh(s)}\poly(n)$ and space $2^{\oeff(s)}\poly(n)$};
	\node[label=right:\rfont{(trivial)}] at (0,-7) {$\Downarrow$};
	\node at (0,-8) {
		There is no algorithm for \pwExample working in\\
		time $2^{\Oh(s)}\poly(n)$ and space $\poly(s,\log n)$};

	\node at (-8.5,0) {
		$\paraNTiSp{f \poly}{f\log} \not\subseteq \paraDTiSp{n^{f+\log}}{f\poly}$};
	\node[label=above:\rfont{(Lemma~\ref{para-strong-iff})}] at (-5,-0) {$\Longleftrightarrow$};

	\node at (-8.5,-2) {$\NL \not\subseteq \DTiSp{2^{\oeff(\log^2 n)}}{n^{\oeff(1)}}$};
	\draw[-implies,double equal sign distance] (-4.5,-0.7) -- (-5.8,-1.5);
	\draw[<-,dashed] (-4.4,-0.4) --node [sloped,midway,above] {\rfont{rescaling}} (-6.2,-1.5);
	
	\draw[rounded corners=2] (-10,-3.95) rectangle (-7.6,-3.4);
	\node at (-8.8,-4.2) {
		Conjecture~\ref{conj:lcs}\\
		\rfont{(Thm~\ref{thm:lcs-completeness})} $\Updownarrow$ \hspace*{4em}\ \\
		$\paraNTiSp{f \poly}{f\log} \not\subseteq \paraDTiSp{n^f}{f\poly}$};
	\node[label=above:\rfont{(Thm~\ref{thm:pwParaCompleteness})}] at (-6.3,-4) {$\Longleftrightarrow$};
	
	
	\node at (-8,-8) {$\NL\not\subseteq\SC$};
	\node[label=above:\rfont{(Lemma~\ref{para-weak-iff})}] at (-6,-8) {$\Longleftrightarrow$};
	\end{tikzpicture}
	\caption{A summary of the relationships between various statements related to Conjecture~\ref{conj:lcs}.}
	\label{fig:summary}
\end{figure}

In this section we relate Conjecture~\ref{conj:lcs} to statements of varying strength concerning complexity class containments, or time-space tradeoffs for pathwidth-constrained problems.
The results are summarized in Figure~\ref{fig:summary}.
Here, we consider algorithms that work uniformly for all values of pathwidth, with complexity depending on both the input size $n$ and the pathwidth $s$ of a given decomposition.

We use the notion of pl-reduction between parameterized problems: it is an algorithm that transforms an instance of one problem with parameter $k$ into an equivalent instance of another problem with parameter $k'\leq f(k)$, working in space $f(k)+\Oh(\log n)$, for some computable function $f$.
Following Elberfeld et al.~\cite{ElberfeldST15} we define\footnote{Throughout this section, in classes $\paraNTiSp{\cdot}{\cdot}$ and $\paraDTiSp{\cdot}{\cdot}$ we always use time and space as the first and the second argument, respectively, hence we drop the subscripts for readability.}
$\paraNTiSp{f\poly}{f\log}$ as the class of parameterized problems that can be solved in non-deterministic time $f(k) \poly(n)$ and space $f(k)\log(n)$ for some computable function $f$, where $k$ is the parameter.
Similarly, $\paraDTiSp{n^f}{f\poly}$ is the class of parameterized problems that can be solved in deterministic time $n^{f(k)}$ and space $f(k)\poly(n)$ for some computable function $f$.
Further classes $\paraDTiSp{t}{s}$ will be defined analogously for different expressions $t,s$.
All those mentioned in the article are closed under pl-reductions.
The reason we do not use the better known fpt-reductions is that $\paraNTiSp{f\poly}{f\log}$ is not expected to be closed under such reductions; its closure under fpt-reductions has been called WNL by Guillemot~\cite{Guillemot11}, but Elberfeld et al.~\cite{ElberfeldST15} argue that a different parameterized class should have this name.

We use the notation $\oeff(h(n))$ as an effective variant of $o(h(n))$; formally, for $f,h:\mathbb{N}\to\mathbb{N}$ we write $f = \oeff(h)$ if there is a non-decreasing, unbounded, computable function $g(n)$ such that $f = \Oh(\frac{h}{g})$.   
The \emph{inverse} of a function $f$ is the function $f^{-1}(n) := \max\{i \mid f(i)\leq n\}$; observe that $f(f^{-1}(n)) \leq n \leq f^{-1}(f(n))$. 
Cai et al.~\cite[Lemmas 3.2, 3.4]{CaiCDF95} showed how computable bounds $f$ and their inverses can be assumed to be computable in appropriately bounded space (logarithmic in $f(n)$) without loss of generality. We use this implicitly when computing tradeoffs in this section, and refer to~\cite{CaiCDF95} for further details.

\subsection{Completeness results and statements equivalent to Conjecture~\ref{conj:lcs}}
Conjecture~\ref{conj:lcs} refers to the following parameterized problem.

\defparproblemu{LCS}{A finite alphabet $\Sigma$, $k$ strings $s_1,s_2,\ldots,s_k$ over $\Sigma$, and an integer $\ell$.}{$k$}{Is there a common subsequence of $s_1,s_2,\ldots,s_k$ of length at least $\ell$?}

We would like to stress that in this variant, the alphabet $\Sigma$ is not of fixed size, but is given on the input and can be arbitrarily large. Under standard fpt-reductions, LCS is known to be W[$t$]-hard for every level $t$~\cite{BodlaenderDFW95}. Moreover, there is some fixed alphabet size for which the problem remains W[1]-hard~\cite{Pietrzak03}, and the problem is also W[1]-complete when parameterized jointly by $k$ and $\ell$~\cite{Guillemot11}. This makes it very hard from the parameterized perspective. From the point of view of pl-reductions, Elberfeld et al.~\cite{ElberfeldST15}, drawing on the work of Guillemot~\cite{Guillemot11}, pinpointed the exact complexity of LCS.

\begin{theorem}[\cite{ElberfeldST15}]\label{thm:lcs-completeness}
LCS is complete for $\paraNTiSp{f\poly}{f\log}$ under pl-reductions.
\end{theorem}

Thus Conjecture~\ref{conj:lcs} is in fact a general statement about parameterized complexity classes. 

\begin{corollary}
Conjecture~\ref{conj:lcs} holds if and only if $\paraNTiSp{f\poly}{f\log} \not\subseteq \paraDTiSp{n^f}{f\poly}$.
\end{corollary}

Similarly as described in the introduction, the best known determinization results (a brute-force approach and Savitch's theorem) imply only that $\paraNTiSp{f \poly}{f\log}$ is contained in classes $\paraDTiSp{n^f}{n^f}$ (commonly known as $\cfont{XP}$) and 
$\paraDTiSp{n^{f(k)\cdot \log n}}{f(k) \cdot \log^2 n}$.
\medskip

To relate parameterized tractability bounds to subexponential bounds, we use the following tool from Cai and Juedes~\cite{CaiJ03}.
For a parameterized problem $\Pi$, its reparametrization (or \emph{extended version}) $\Pi^{\log n}$ is defined as the same problem parameterized by $s/\lg n$, where $s$ was the old parameter.
In particular, \pwExample{}$^{\log n}$ is the following parameterized problem:

\defparproblemu{\pwExample{}$^{\log n}$}{A graph $G$ with a path decomposition of width $k \cdot \lg n$}{$k$}{Is $G$ 3-colorable?}

Similarly as in Theorem~\ref{thm:pwCompleteness}, pathwidth-constrained problems turn out to be complete for non-deterministic computation with simultaneous time and space bounds.

\begin{theorem}\label{thm:pwParaCompleteness}
\pwExample{}$^{\log n}$ is complete for \paraNTiSp{$f \poly$}{$f\log$} under pl-reductions.
\end{theorem}
\begin{proof}
To show containment, we give a non-deterministic algorithm solving \pwExample in time $\poly(n)$ ans space $\Oh(s+\log n)$, where $s$ is the pathwidth of the graph decomposition given on the input.
The algorithm proceeds on consecutive bags of the pathwidth decomposition by guessing the color of every vertex when it is introduced for the first time and comparing it with previously guessed colors of adjacent vertices in the bag.
At any moment, only the colors of all vertices in the current bag and the position in the decomposition need to be remembered, hence $\Oh(s + \log n)$ space suffices.
Correctness follows from the fact that every color is guessed exactly once and every edge is checked, because every pair of adjacent vertices is contained in some bag.
Since $k=\frac{s}{\lg n}$ is the parameter of the rescaled problem,
the same algorithm works in time $\poly(n)$ and space $\Oh(k \log n + \log n)$, and hence in time $f(k) \poly(n)$ and space $\Oh(f(k) \log n)$ for $f(k)=k$.

To show hardness, consider any non-deterministic Turing machine $M$ solving a problem with input size $n$ and parameter $k$ in time $t(n,k)=f(k) \poly(n)$ and space $s(n,k)=\Oh(f(k) \log n)$, for some computable function $f$.
Given an instance $\alpha$ of this problem with size $n$ and parameter $k$, the reduction of Lemma~\ref{lem:compGadget}, with $h=0$, $t=t(n,k)$, $s=s(n,k)$, will  output a CNF-SAT instance $\varphi$, using $\Oh(\log (n +s+t))) = \Oh(\log f(k) + \log n)$ space. 
We can easily add clauses on the named variables to enforce the correct shape of the initial configuration and that the final configuration is accepting. Hence, the formula can be modified so that $\varphi$ is satisfiable if and only if $M$ accepts $\alpha$.
The lemma further provides a path decomposition of $\varphi$'s primal graph of width $\Oh(s)=\Oh(f(k) \log n)$; it is easy to verify that the additional clauses do not spoil this bound.
Using Lemma~\ref{lem:equivalenceOfProblems} we can then reduce this instance to an instance of $\pwExample$ with a path decomposition of width $\Oh(f(k) \log n)$, that is, to \pwExample{}$^{\log n}$ with parameter $k'=\Oh(f(k))$.
\end{proof}

Conjecture~\ref{conj:lcs} is thus equivalent to the statement that $\pwExample{}^{\log n}$ is not in $\paraDTiSp{n^f}{f\poly}$, in other words, that $\pwExample$ cannot be solved deterministically in time $n^{f(s/\lg n)}$ and space $f(s/\lg n) \poly(n)$ for any computable function $f$, where $s$ is the width of the input path decomposition.
To contrast pathwidth with tree-depth, one may easily observe (see Lemma~\ref{lem:deter-3col}) that an instance of {\textsc{3Coloring}} with a tree-depth decomposition of depth $s$ can be solved deterministically using $\Oh(s+\log n)$ space. This places
\tdExample{}$^{\log n}$ in $\paraDTiSp{n^f}{f \log}$, a class usually known as $\cfont{XL}$.

\subsection{Statements weaker than Conjecture~\ref{conj:lcs}}
Similarly as in the work of Cai and Juedes~\cite{CaiJ03}, we show that the parameterized complexity of the reparameterized problem $\pwExample{}^{\log n}$ is related to subexponential bounds in the complexity of $\pwExample{}$. Recall that we always assume that instances of $\pwExample{}$ come with appropriate path decompositions of the graph.

\begin{lemma}\label{para-weak-sandwich}
Assuming Conjecture~\ref{conj:lcs}, there is no algorithm for \pwExample working in time $2^{\Oh(s)} \poly(n)$ and space $2^{\oeff(s)}  \poly(n)$ (for all values of pathwidth $s$).
\end{lemma}
\begin{proof}
Suppose to the contrary that \pwExample can be solved in time $2^{\Oh(s)} \poly(n)$ and space $2^{\oeff(s)}  \poly(n)$. We show that $\paraNTiSp{f \poly}{f\log} \subseteq \paraDTiSp{n^f}{f\poly}$, contradicting Conjecture~\ref{conj:lcs}.

The assumption implies that \pwExample{}$^{\log n}$ can be solved in time $2^{\Oh(k \cdot \lg n)} = n^{\Oh(k)}$ and space $2^{k \cdot \lg n / g(k \cdot \lg n)}$ for some unbounded and non-decreasing computable function $g(\cdot)$.
If $k\leq g(k \cdot \lg n)$, then the bound on space is bounded by $n$.
Otherwise, if $k > g(k \lg n) \geq g (\lg n)$, then $n \leq 2^{g^{-1}(k)}$.
In this case the bound on space is bounded by a computable function of $k$, namely $2^{k\cdot g^{-1}(k)}$.
Hence in each case, the same algorithm solves \pwExample{}$^{\log n}$ in time $n^{\Oh(k)}$ and space $n+2^{k\cdot g^{-1}(k)}$. By Theorem~\ref{thm:pwParaCompleteness}, this implies $\paraNTiSp{f \poly}{f\log} \subseteq \paraDTiSp{n^f}{f\poly}$.
\end{proof}

An even weaker statement is equivalent to $\NL \not\subseteq \SC$ by the following padding argument.

\begin{lemma}\label{para-weak-iff}
There is no algorithm for \pwExample working in time $2^{\Oh(s)}\poly(n)$ and space $\poly(s,\log n)$ (for all values of pathwidth $s$) if and only if $\NL \not\subseteq \SC$.
\end{lemma}
\begin{proof}
If there was such an algorithm, then it would solve $\pwEx{\log n}$ in polynomial time and polylogarithmic space. However, from Theorem~\ref{thm:pwCompleteness} it follows that $\pwEx{\log n}$ is complete for $\NL$, so this would imply that $\NL \subseteq \SC$.

For the other direction, suppose $\NL \subseteq \SC$.
Then, $\pwExample$ on instances with path decompositions of width at most $\lg n$ can be solved in polynomial time and space $\Oh(\log^c n)$ for some constant $c$.
Thus, an instance of size $n$ and a path decomposition of width $s$ can be padded to size $n' = \max(n, 2^s)$ (by adding isolated vertices --- neither the answer nor the width changes) and solved by this algorithm in time $\poly(n')$ and space $\Oh(\log^c n')$.
This solves  $\pwExample$ in time at most $\max(\poly(n),\poly(2^s)) \leq 2^{\Oh(s)}\poly(n)$ and space $\max(\log^c n, s^c) \leq \poly(s,\log(n))$.
\end{proof}

\subsection{Statements stronger than Conjecture~\ref{conj:lcs}}
Contrary to the results of Cai and Juedes~\cite{CaiJ03}, in our context we are unable to prove the converse of Lemma~\ref{para-weak-sandwich}.
However, we can get arbitrarily close to it, in a sense.

\begin{lemma}\label{para-strong-sandwich}
If Conjecture~\ref{conj:lcs} fails, then for every arbitrarily slowly growing, unbounded, computable function $g$, \pwExample can be solved in time $2^{s\cdot g(s)} \poly(n)$ and space $2^{\oeff(s)} \poly(n)$ (for all values of pathwidth $s$).
\end{lemma}
\begin{proof}
The assumption is equivalent to $\paraNTiSp{f \poly}{f\log} \subseteq \paraDTiSp{n^f}{f\poly}$, which in turn implies that \pwExample{} can be solved in time $n^{f(s/\lg n)}$ and space $f(s/\lg n) \cdot \poly(n)$ for some computable, increasing function $f$.
Let $g(s)$ be an arbitrarily slowly growing, unbounded, computable function.
Without loss of generality assume that $g$ is non-decreasing and $g(s)\leq s$.
Let $n'=2^{s / f^{-1}(g(s))}$.
Observe\footnote{A careful reader probably noticed that some floors/ceilings are formally necessary here. For the sake of readability, in this and other proofs in this section we ignore such straightforward details, as their introduction would only obfuscate the main ideas.} that
$f(s/\log n') = f(f^{-1}(g(s))) \leq g(s)$ and $f^{-1}(g(s))$ is an unbounded, computable function of $s$.

Consider an instance of \pwExample{} of size $n$ with a path decomposition of width $s$.
If $s \leq \lg n$, then the assumed algorithm runs in time $n^{f(s/\lg n)} \leq n^{f(1)} = \poly(n)$ and space $f(1) \cdot \poly(n) = \poly(n)$.
If $\lg n < s \leq f^{-1}(g(s)) \cdot \lg n$, then $f(s/\lg n)\leq g(s)$ and hence the assumed algorithm runs in time $n^{f(s/\lg n)} \leq n^{g(s)} = 2^{\lg n \cdot g(s)} \leq 2^{s \cdot g(s)}$ and space $f(s/\lg n) \cdot \poly(n) \leq g(s) \cdot \poly(n) = \poly(n)$ (the last equality follows from $g(s)\leq s$ and $s\leq n$).
Finally, if $s > f^{-1}(g(s)) \cdot \lg n$, then $n' > n$, hence we can pad the instance (by adding isolated vertices  --- neither the answer nor the width changes) to size $n'$. Applying the assumed algorithm to the padded instance solves the problem in time
 ${n'}^{f(s/\lg n')} \leq 2^{s \cdot \frac{g(s)}{f^{-1}(g(s))}} \leq 2^{s \cdot g(s)}$
and space 
$f(s/\lg n') \cdot {n'}^{\Oh(1)} \leq g(s) \cdot  2^{\Oh(s / f^{-1}(g(s)) )} = 2^{\oeff(s)}$.
\end{proof}


For a somewhat less natural, stronger variant of Conjecture~\ref{conj:lcs}, we can show a similar, but exact correspondence (note the quasi-polynomial factor on both sides).
The proof is very similar.

\begin{lemma}\label{para-strong-iff}
There is no algorithm for $\pwExample$ working in time $2^{\oeff(s^2)}n^{\Oh(\log n)}$ and space $2^{\oeff(s)}\poly(n)$ (for all values of pathwidth $s$) if and only if
$\paraNTiSp{f\poly}{f\log} \not\subseteq \paraDTiSp{n^{f+\log}}{f\poly}$.
\end{lemma}
\begin{proof}
Suppose first there is an algorithm for $\pwExample$ working in time $2^{s^2/g(s)}n^{\Oh(\log n)}$ and space $2^{s/g(s)}\poly(n)$, for some unbounded, non-decreasing, computable function $g$.
Consider an instance of $\pwExample{}^{\log n}$ of size $n$, parameter $k$, and hence equipped with a path decomposition of the graph of width $k\cdot \lg n$.
If $k^2 < g(\lg n)$, then the algorithm solves the instance in time at most $2^{\frac{k^2\lg^2 n}{g(k \lg n)}}n^{\Oh(\log n)} \leq 2^{\Oh(\log^2 n)}$ and space $2^{\frac{k\lg n}{g(k\lg n)}} \leq \poly(n)$.
Otherwise, if $g(\lg n) \leq k^2$, then $n$ is bounded by a computable function of $k$, namely $2^{g^{-1}(k^2)}$.
Hence the algorithm solves the instance in time and space bounded by a computable function $f$ of $k$.
Therefore in any case, the same algorithm solves $\pwExample{}^{\log n}$ in time $f(k)+2^{\Oh(\log^2 n)}$ and space $f(k)+\poly(n)$, which implies $\paraNTiSp{f\poly}{f\log} \subseteq \paraDTiSp{n^{f+\log}}{f\poly}$.

For the converse, suppose now that $\paraNTiSp{f\poly}{f\log} \subseteq \paraDTiSp{n^{f+\log}}{f\poly}$.
Then, there is an algorithm for $\pwExample$ working in time $n^{f(s/\lg n) + \Oh(\log n)}$ and space $f(s/\lg n) \poly(n)$ on instances of size $n$ and with path decompositions of width $s$, for some computable, increasing function $f$.
Given such an instance of $\pwExample$, let $n'=2^{s/f^{-1}(\lg n)}$.
If $n' < n$, then $\frac{s}{\lg n} < f^{-1}(\lg n)$, hence $f(s/\lg n) < f(f^{-1}(\lg n)) \leq \lg n$ and the assumed algorithm works in time $n^{\Oh(\log n)}$ and space $\poly(n)$.
Otherwise, if $n' \geq n$, then we can pad the instance to size $n'$ (by adding isolated vertices --- neither the answer nor the width changes).
Then, since $\lg n'=s/f^{-1}(\lg n)\leq s/f^{-1}(\lg s)=\oeff(s)$,
the padded instance is solved in time $2^{\lg n' \cdot \left(\lg n + \Oh(\lg n')\right)} = 2^{\Oh(\lg^2 n')} = 2^{\oeff(s^2)}$ and space $\lg n \cdot \poly(n') = 2^{\Oh(\lg n')} = 2^{\oeff(s)}$.
\end{proof}



\subsection{A summary}
Theorem~\ref{thm:main} follows from Lemmas~\ref{para-weak-sandwich} and~\ref{para-strong-sandwich}. We summarize the relationships around Conjecture~\ref{conj:lcs} in Figure~\ref{fig:summary}. 
The weakest statement there is $\NL \not\subseteq \SC$, a widely explored hypothesis in complexity theory.
Since {\textsc{Directed $(s,t)$-Reachability}} (asking given a directed graph and two nodes $s,t$, whether is $t$ reachable from $s$) is an \NL-complete problem, this is also equivalent to the question of whether this problem can be decided in polynomial time and polylogarithmic space simultaneously.
However, even this weakest statement is not known to be implied by better known conjectures such as the Exponential Time Hypothesis.
It seems that the simultaneous requirement on bounding two complexity measures---time and space---has a nature that is independent of the usual time complexity considerations. Hence, new complexity assumptions may be needed to explore this paradigm, and we hope that Conjecture~\ref{conj:lcs} may serve as a transparent and robust example of such.

In a certain restricted computation model (allowing operations on graph nodes only, not on individual bits), unconditional tight lower bounds have been proved by Edmonds et~al.~\cite{EdmondsPA99}:
it is impossible to decide {\textsc{Directed $(s,t)$-Reachability}} in time $2^{o(\log^2 n)}$ and  space $\Oh(n^{1-\varepsilon})$ (for any $\varepsilon>0$), even if randomization is allowed.
Essentially all known techniques for solving {\textsc{Directed $(s,t)$-Reachability}} are known to be implementable in this model~\cite{LuZPC05} (including Depth- and Breadth-First Search, as well as the well-known theorems of Savitch, of Immerman and Szelepcs\'{e}nyi, and Reingold's breakthrough),
therefore this strongly suggests that no algorithm running in time $2^{\oeff(\log^2 n)}$ and space $n^{\oeff(1)}$ is possible, that is, $\NL \not\subseteq \DTiSp{2^{\oeff(\log^2 n)}}{n^{\oeff(1)}}$.

By Theorem~\ref{thm:pwCompleteness}, this is equivalent to saying that $\pwEx{\log}$ cannot be solved in these time and space bounds.
The strongest statement on Figure~\ref{fig:summary} is a rescaling of this, that is, it implies $\NL \not\subseteq \DTiSp{2^{\oeff(\log^2 n)}}{n^{\oeff(1)}}$ by a trivial padding argument, but the reverse implication is also probable in the sense that any proof of the latter would likely scale to prove the former.
However, it is still possible that an algorithm working in polynomial space refutes the stronger statement even though $\NL \not\subseteq \DTiSp{2^{\oeff(\log^2 n)}}{n^{\oeff(1)}}$.

\pagebreak
\section{Treedepth}\label{sec:treedepth}
\newcommand{\dum}{\blacklozenge}
\newcommand{\blsize}{\mathbf{b}}

\subsection{Characterization via NTMs with a small auxilliary stack}
In this section we prove a completeness theorem for small tree-depth computations, i.e., Theorem~\ref{thm:main-td}. Let $s: \mathbb{N} \to \mathbb{N}$ be a nice function.
Before we proceed, we discuss more precisely the model of machines used to define the class $\NAuxSA{\poly}{\log}{s}$. The machine has three tapes: a read-only input tape, a working tape of length $\Oh(\log n)$, and a stack tape of length $s(n)$. Each tape contains symbols from some fixed, finite alphabet $\Sigma$. Initially both the working tape and the stack tape are empty, i.e., filled with blank symbols. On each of the tapes there is a head, and the transitions of the machine depend on its state and the triple of symbols under the heads. The access restraints to each of the tapes are as follows: The input tape is a read-only tape. The working tape can be both read and written on by the machine. The stack tape can be read but not freely written on; instead, the transitions of the machine may contain instructions of the form $\texttt{push}\, \sigma$ or $\texttt{pop}$, where $\sigma$ is some non-blank symbol of $\Sigma$. In case of $\texttt{push}\, \sigma$, the first blank symbol of the tape is replaced by $\sigma$. In case of $\texttt{pop}$, the last non-blank symbol of the tape is replaced by a blank. Since $s$ is nice, $s(n)\leq \poly(n)$, so within the working tape the machine can keep track of the current height of the stack and the index on which the stack's head is positioned. The machine accepts if it can reach an accepting state through a sequence of transitions, and for a problem in $\NAuxSA{\poly}{\log}{s}$, there must always be an accepting run where the number of transitions is bounded by a polynomial in $n$.

In this section we show that, for any nice function $s(n)\geq \log^2 n$, \tdEx{s} is complete for $\NAuxSA{\poly}{\log}{s}$ and $\NAuxSA{\poly}{s/\log}{s}$ under logspace reductions. In particular, this implies that these two classes are equal. We start by showing containment, exemplifying how the resources are used.

\begin{lemma}\label{lem:tdContain}
For any nice function $s(n)$, $\tdEx{s}$ is in $\NAuxSA{\poly}{\log}{s}$.
\end{lemma}
\begin{proof}
The input consists of a graph $G$ and a mapping of its vertices into a rooted forest $\mathcal{T}$; we assume any natural encoding for which validity of the decomposition can be checked in logarithmic space. 
We now give a non-deterministic algorithm showing that \tdEx{s} belongs to $\NAuxSA{\poly}{\log}{s}$. The algorithm considers the trees of $\mathcal{T}$ one by one, and for each tree $T$ of $\mathcal{T}$ it explores $T$ by a depth-first search beginning from the root. Since the trees are rooted and for each vertex we store its parent and the list of its children, it suffices to maintain only the identifier of the current vertex during the search. When entering a node, the algorithm non-deterministically guesses its color (encoded using a constant number of bits) and pushes it onto the stack. When the depth-first search withdraws from a vertex to its parent, the algorithm pops its color from the stack. Thus, the stack always contains the list of guessed colors of vertices on the path from the current vertex to the root of its tree. To verify the correctness of the coloring, after guessing the color of some vertex $u$, we check that all its ancestors in $\mathcal{T}$ that are adjacent to it have different colors; this information can be retrieved from the stack. In this way, the color of each vertex is guessed exactly once, and for each edge of the tree we verify that its endpoints have different colors when considering the endpoint that is lower in $\mathcal{T}$. Therefore, the machine accepts if and only if the graph has a proper 3-coloring.
\end{proof}


Clearly when $s(n)\geq \log^2 n$, $\NAuxSA{\poly}{\log}{s} \subseteq \NAuxSA{\poly}{s/\log}{s}$.
The next step is to show how the  stack operations of the latter class' machines can be regularized. This idea originates in the approach of Akatov and Gottlob~\cite{AkatovG10}.
Following their ideas, we define a regular stack machine in the following way.
For any valid sequence $S$ of push and pop operations that starts and ends with an empty stack, define the corresponding push-pop tree $\tau(S)$ to be the ordered tree (a tree with an order imposed on the children of each node) in which a depth-first search would result in the sequence $S$; specifically,
\begin{itemize}
\item $\tau(\varepsilon)$ is a single root node,
\item $\tau(\texttt{push}\ S\ \texttt{pop})$ is a new root node with $\tau(S)$ attached as the only child subtree,
\item $\tau(S_1 S_2)$ is obtained from $\tau(S_1)$ and $\tau(S_2)$ by identifying their roots, and putting all the children of the root of $\tau(S_1)$ before all children of the root of $\tau(S_2)$.
\end{itemize}

We say that a language is in $\text{reg-}\NAuxSA{\poly}{s/\log}{s}$ if it is recognized by an NTM with $s(n)/\log(n)$ working space and an auxiliary stack of height $s(n)$ that has the following properties:
\begin{enumerate}[label=(\alph*)]
\item\label{r:blocks} The machine pushes and pops blocks of $\lceil s(n)/\lg(n) \rceil$ symbols at a time. More precisely, the machine uses a pre-specified block size $\blsize=\lceil s(n)/\lg(n) \rceil$. The reader may imagine that the $\texttt{push}$ operation causes a simultaneous push of a block of $\blsize$ symbols stored on the working tape; say, from the first $\blsize$ positions of the working tape. The $\texttt{pop}$ operation causes a simultaneous pop of a block of $\blsize$ symbols from the stack (i.e., replacing them with blank symbols on the stack tape). 
\item\label{r:conf} Whenever the machine decides to push or pop, it can only change its state. That is, the heads cannot move and the content of the working tape does not change. Moreover, the decision about using a push or pop transition is done solely based on the machine's state, that is, such transitions are available if and only if the machine's state belongs to some subset of states, independently of the symbols under the machine's heads.
\item\label{r:binary} If the machine accepts an input $\alpha$, then there is a run on $\alpha$ where the corresponding push-pop tree (where each operation of pushing/popping a block is considered atomic) is the full binary tree of depth exactly $c \lceil \lg n\rceil$, for some fixed integer $c$. In particular, at the moment of accepting the stack is empty.
\end{enumerate}
Obviously, the block pushes and pops described in restriction~\eqref{r:blocks} can be simulated by a standard machine in $\Oh(\blsize)$ steps, so their introduction does not give additional expressive power to the computation model. 

Restriction~\eqref{r:conf} is a technical adjustment that will be used to streamline future constructions. Intuitively, restriction~\eqref{r:blocks} is easy to achieve, because the machine has enough working space to simulate the top $\lceil s(n)/\lg(n) \rceil$ symbols from the stack on the working tape, and group pushes and pops into blocks of size $\blsize$. The most important restriction is~\eqref{r:binary}: the push-pop tree has a fixed shape of a full binary tree. This property will be essential when reducing an arbitrary problem from $\NAuxSA{\poly}{s/\log}{s}$ to \tdEx{s}, because the push-pop tree will form a ``skeleton'' of the graph output from the reduction. In order to achieve this property, we use the following technical result of Akatov and Gottlob~\cite{AkatovG10,akatov2010exploiting}, which was also used independently by Elberfeld et al.~\cite{ElberfeldJT10}.
The \emph{traversal ordering} of the nodes of an ordered tree is the linear ordering which places a parent before its children and, for two children $a,b$ of a node, $a$ occurring before $b$, places all descendants of $a$ before all descendants of $b$. In other words, the traversal ordering is the ordering in which the nodes are visited in a depth-first search started at the root.

\begin{lemma}[Lemma 3.3 of \cite{akatov2010exploiting}; Theorem 3.14 of \cite{ElberfeldJT10}]\label{lem:treeEmbedding}
Given an ordered tree $T$ with $n$ nodes and depth at most $\lg n$, one can in logarithmic space compute an embedding (an injection that preserves the ancestor relation and traversal ordering) into a full binary tree of depth $4 \lg n$.
\end{lemma}

We are now ready to prove that every problem from $\NAuxSA{\poly}{s/\log}{s}$ can be recognized by a regularized machine, in a similar way as was the case in~\cite{AkatovG10}.

\begin{lemma}\label{lem:tdRegularize}
$\NAuxSA{\poly}{s/\log}{s} \subseteq \text{reg-}\NAuxSA{\poly}{s/\log}{s}$.
\end{lemma}
\begin{proof}
Consider a machine $M$ placing $L$ in $\NAuxSA{\poly}{s/\log}{s}$.
We modify $M$ to comply with restrictions~\eqref{r:blocks},~\eqref{r:conf},~and~\eqref{r:binary}.

First, to achieve restriction~\eqref{r:blocks}, we designate the first $\blsize=\lceil s(n)/\lg(n) \rceil$ of the working tape as a buffer for simulating the top of the stack. Whenever during the run $M$ has $p$ symbols on the stack, then after modification the top-most $p\mod \blsize$ symbols are stored in the buffer, while the remaining symbols are stored in $\lfloor p/\blsize\rfloor$ blocks on the stack. The operations on the stack are simulated in the buffer. Whenever the size of the buffer reaches full size $\blsize$, the modified machine invokes the block $\texttt{push}$ operation and clears the buffer. Whenever $M$ wants to pop a symbol but the buffer is empty, the modified machine copies the top $\blsize$ symbols from the stack to the buffer, invokes the block $\texttt{pop}$ operation, and then simulates the pop of $M$ in the buffer. It is straightforward to simulate the read access to the stack tape with a polynomial-time overhead in the running time.

To achieve restriction~\eqref{r:conf}, whenever the machine would like to push or pop, we split this transition into three. In the first transition, the machine only verifies the symbols under the heads, and enters a state ``ready to push/pop'' where it additionally remembers the target state and additional operations (moving the heads) to be performed after the stack operation. Then, in the second transition, it performs only the push or pop; note that this transition does not depend on the symbols under the heads, and the heads also do not move. Finally, in the third transition the machine performs the remembered head movements and reaches the target state.

Finally, we concentrate on restriction~\eqref{r:binary}. Let us add a dummy symbol $\dum$ to the alphabet. Suppose that on some input $\alpha$ of length $n$ the machine $M$ (after the modifications above) has an accepting run. By a simple modification, we may assume that $M$ always empties the stack before accepting. Since the maximum stack height is still $\Oh(s(n))$ and stack operations are done in blocks of $\blsize=\lceil s(n)/\lg n \rceil$ operations at a time, the corresponding push-pop tree $T$ (where blocks operations are considered as atomic) has depth at most $c_1\lg n$ and size at most $n^{c_2}$ for some integers $c_1,c_2$. By Lemma~\ref{lem:treeEmbedding}, there exists an integer $c$ for which $T$ can be always embedded into a tree $T_0$ that is a full binary tree of depth exactly $c\lceil \lg n\rceil$. We now modify $M$ so that it has also an accepting run whose push-pop tree is $T_0$.

We modify machine $M$ in the following manner. First, $M$ will keep track of the current position of the computation in the push-pop tree $T_0$, encoded as binary string of length at most $c\lceil \lg n\rceil$. We add the possibility for $M$ to guess non-deterministically, at any moment of the computation, to push a block of $\blsize$ dummy symbols without changing the machine's configuration, apart from updating the current position in $T_0$ (if the computation is in a leaf of $T_0$, then this transition is not allowed). Similarly, $M$ can, at any moment, pop the top-most block of the stack provided it consists only of dummy symbols, update the position in $T_0$, without changing the configuration otherwise. The read access to the stack tape is simulated by ignoring the dummy symbols, i.e., the head always continues browsing the tape until a non-dummy symbol is found. Since the dummy symbols are effectively ignored by the computation, it is easy to see that the modified machine accepts if and only if the original one accepts.

Take now an accepting run of $M$ on $\alpha$ and consider an embedding $\eta$ of its push-pop tree $T$ into $T_0$. Construct a run of modified $M$ on $\alpha$ by adding non-deterministic pushes and pops of blocks of dummy symbols $\dum$ for all nodes of $T_0$ that are not in the image of $\eta$. Thus, the push-pop tree of the modified run is exactly $T_0$. 
\end{proof}



Knowing that computations for NAuxSA can be conveniently regularized, we can describe the existence of such a computation by a CNF formula ``wrapped around'' the rigid shape of the full binary tree that encodes the push-pop tree of the run. 
This was also the idea of Akatov and Gottlob~\cite{AkatovG10}, but our reduction needs to introduce many more elements, in particular copies of the gadget of Lemma~\ref{lem:compGadget}, due to a less expressive target language.

\begin{lemma}\label{lem:tdHardness}
Suppose $L \in \text{reg-}\NAuxSA{\poly}{s/\log}{s}$.
Then $L\lred\text{td-CNF-SAT}[s]$.
\end{lemma}
\begin{proof}
Let $M$ be an appropriate machine recognizing $L$; by the assumption, $M$ satisfies restrictions~\eqref{r:blocks},~\eqref{r:conf},~and~\eqref{r:binary}. Consider an input word $\alpha$; let $n=|\alpha|$.
Let $T$ be the push-pop tree of $M$ on input $\alpha$. By definition $T$ is the full binary tree of depth exactly $c \lceil\lg n\rceil$ for some fixed integer $c$ that depends on $L$ only. We assume that the machine works in time at most $t(n)$ and uses work tape of size at most $s(n)/\log n$.

Before we proceed to the formal description, let us elaborate on the intuition; see Figure~\ref{fig:computation} for a visualization. An Euler tour of the nodes and edges of $T$ corresponds to subsequent phases of $M$'s execution on $\alpha$. We think of the computation as starting at the root node, moving down an edge whenever a push is made and moving up an edge whenever a pop is made. Since $T$ is fixed, the idea is to construct a gadget for each node of $T$ and to wire the gadgets so that they encode the full computation of $M$. Each node gadget will contain three copies of the computation gadget of Lemma~\ref{lem:compGadget}. These copies respectively encode the parts of the computation before going into the first subtree, between withdrawing from them first subtree and proceeding to the second subtree, and after withdrawing from the second subtree. Each part of the computation depends on symbols pushed onto the stack on the path to the root, but is independent of the guesses in different branches. This, together with the bound of Lemma~\ref{lem:compGadget} on the tree-depth of the computation gadget, will give rise to a natural tree-depth decomposition of depth $\Oh(s(n))$ of the obtained graph.

\begin{figure}[htbp!]
                \centering
                \includegraphics[width=\textwidth]{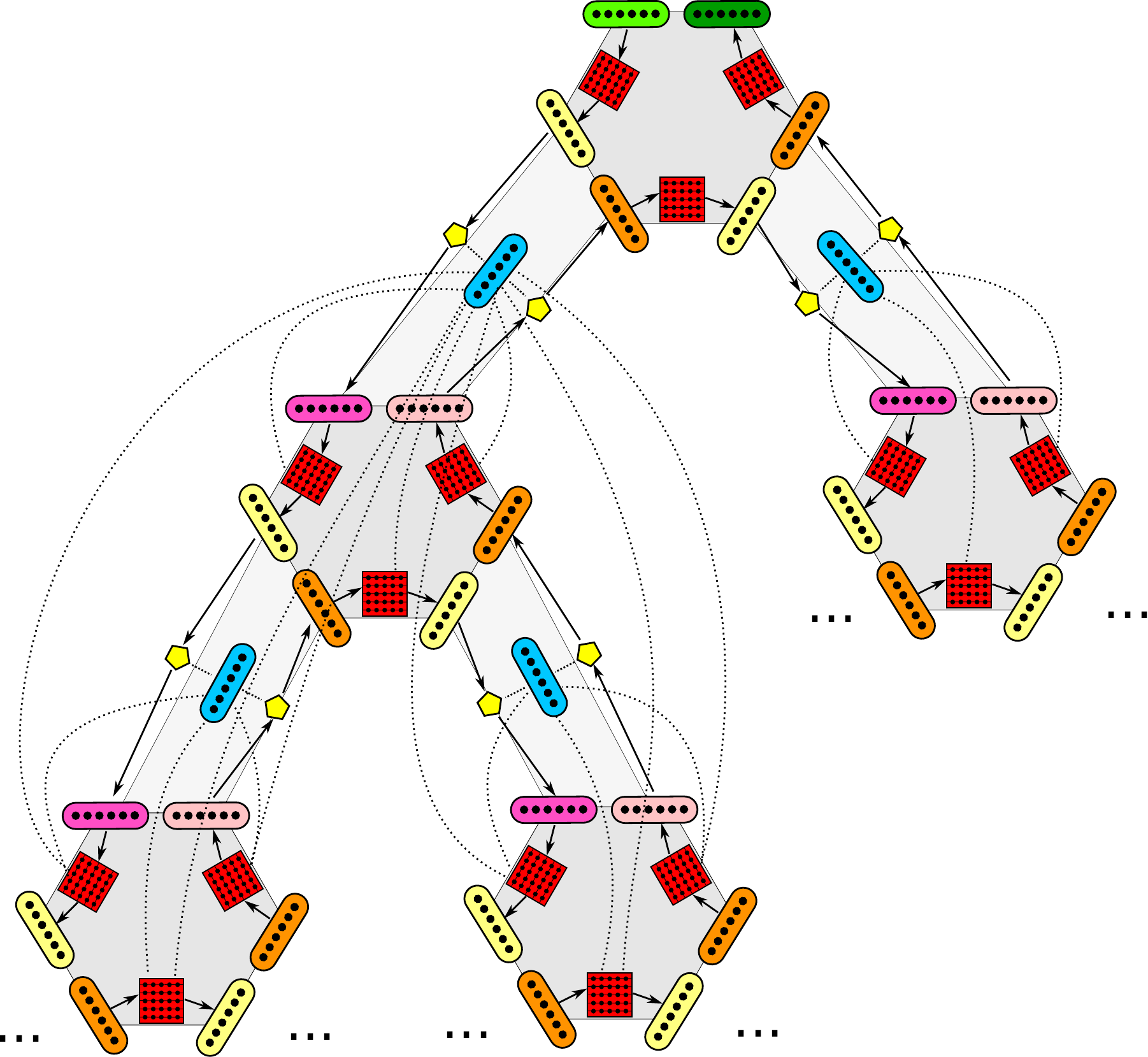}
\caption{The construction of Lemma~\ref{lem:tdHardness}. Blocks $\texttt{stack}(q)$ are depicted in light blue, blocks $\texttt{before-push}(e)$ are depicted in yellow, blocks $\texttt{after-push}(e)$ are depicted in violet, blocks $\texttt{before-pop}(e)$ are depicted in pink, blocks $\texttt{after-pop}(e)$ are depicted in orange, block $\texttt{init}$ is depicted in light green, and block $\texttt{final}$ is depicted in dark green.
Copies of the computation gadget are shown as red squares, while clauses validating push and pop operations are depicted as yellow pentagons. }\label{fig:computation}
\end{figure}

We proceed to the formal description; the reader is advised to look at Figure~\ref{fig:computation} while reading. For each edge $e$ of tree $T$ we create a block $\texttt{stack}(e)$ of $\Oh(\blsize)$ variables, describing the values pushed on the stack when accessing the lower endpoint of $e$ from the upper during the computation.
Moreover, we create four blocks of $\Oh(\frac{s(n)}{\log n})$ variables to describe the configurations immediately before and after the corresponding push and pop operations concerning block $\texttt{stack}(e)$: $\texttt{before-push}(e)$, $\texttt{after-push}(e)$, $\texttt{before-pop}(e)$ and $\texttt{after-pop}(e)$, respectively. Each of these blocks contains the information about (a) the full content of the working tape, (b) the positions of heads on all the three tapes (encoded as indices in binary), and (c) the machine's state. Similarly, we create also two blocks of $\Oh(\frac{s(n)}{\log n})$ variables describing the initial and final configurations: $\texttt{init}, \texttt{final}$. We enforce them to be the initial and accepting configurations, respectively. For $\texttt{init}$ this only requires introduction of a number of clauses of size $1$ that precisely describe the initial configuration. For $\texttt{final}$, we only need to verify that the final state is accepting. Since the description of the state uses a constant number of variables, this can be easily done by imposing a constant number of CNF clauses on them.

We now wire these groups of variables in order to simulate the machine's execution.

First, for each edge $e$ of $T$, we connect $\texttt{before-push}(e)$ with $\texttt{after-push}(e)$ and $\texttt{before-pop}(e)$ with $\texttt{after-pop}(e)$ by clauses that enforce that valid push/pop transitions are being used at these points. By restriction~\eqref{r:conf}, during these transitions only the machine's state can change. Hence, we can enforce that in blocks $\texttt{before-push}(e)$ and $\texttt{after-push}(e)$ all the information apart from the machine's state is exactly the same, and the same holds also for $\texttt{before-pop}(e)$ and $\texttt{after-pop}(e)$; this can be easily done using clauses of size $2$. Then, we need to verify that the performed transition was indeed available. Since the availability of a push/pop transition depends solely on the machine's state (restriction~\eqref{r:conf}), we can verify that the original and target state can be connected by a push/pop transition using a constant-size family of CNF clauses on the variables describing these states. Finally, we need to make sure that the values memorized in block $\texttt{stack}(e)$ are indeed the ones that the machine intended to push onto the stack. Therefore, in the connection between $\texttt{before-push}(e)$ and $\texttt{after-push}(e)$ we also verify, again using clauses of size $2$, that values stored in $\texttt{stack}(e)$ are exactly the same as values on the $\blsize$ first positions of the working tape. Observe that no such check is needed between $\texttt{before-pop}(e)$ and $\texttt{after-pop}(e)$.


Then, we connect (in a manner described later) pairs of configurations that correspond to the beginning and end of a computation without stack operations.
That is, for the root note $q$ with children $l,r$ we connect:
$\texttt{init}$ with $\texttt{before-push}(ql)$,
$\texttt{after-pop}(ql)$ with $\texttt{before-push}(qr)$
and $\texttt{after-pop}(qr)$ with $\texttt{final}$.
For every internal, non-root node $q$ with parent $p$ and children $l,r$, we connect: $\texttt{after-push}(pq)$ with $\texttt{before-push}(ql)$, 
$\texttt{after-pop}(ql)$ with $\texttt{before-push}(qr)$
and $\texttt{after-pop}(qr)$ with $\texttt{before-pop}(pq)$.
For every leaf node $q$ with parent $p$ we connect 
$\texttt{after-push}(pq)$ with $\texttt{before-pop}(pq)$.

Each connection between two blocks is made by creating a new copy of the computation gadget of Lemma~\ref{lem:compGadget}, where:
\begin{itemize}
\item $t$ from the statement of Lemma~\ref{lem:compGadget} is simply $t(n)$.
\item $s$ from the statement of Lemma~\ref{lem:compGadget} is $\frac{s(n)}{\lg n}$.
\item $h$ is equal to the stack height at the corresponding moment. That is, this height is exactly $\blsize$ times the number of edges on the path from the root to the current node $q$ of the push-pop tree $T$ (note that $h\leq \blsize \cdot c\lceil \lg n\rceil\leq 4c \cdot s(n)$).
\item The input word is $\alpha$.
\end{itemize}
Variables $u_1,\dots,u_{s'}$ and $v_1,\dots,v_{s'}$ of the gadget are identified with the two variable blocks we connect.
Variables $w_1,\dots,w_{h'}$ of the gadget are identified with consecutive variables from $\texttt{stack}(e)$ for all edges $e$ on the path from root to the current node $q$. Note that in this manner, each computation gadget has the input string $\alpha$ encoded within, so that the computation simulated by the gadget has read access to $\alpha$.

This concludes the construction of the formula.
By Lemma~\ref{lem:compGadget} and restriction~\eqref{r:binary}, it is clear that the formula has a satisfying assignment iff there is an accepting run of $M$.
It remains to construct a tree-depth decomposition of the formula's primal graph $G$ that has depth $\Oh(s(n))$.

Let $R$ be the subset of variables consisting of all the variables contained in the named blocks (i.e. $\texttt{init}$, $\texttt{stack}(e)$, $\texttt{before-push}(e)$, etc.).
First, we create a mapping $\eta$ that maps variables of $R$ to the nodes of $T$.
We begin by mapping the variables of $\texttt{init}$ and $\texttt{final}$ blocks to the root of $T$.
For every non-root node $q$ of $T$ with parent $p$, we map the variables of $\texttt{stack}(pq)$ to $q$.
Additionally, if $q$ has a parent $p$ or children $l,r$, we map the variables of $\texttt{after-push}(pq)$, $\texttt{before-push}(ql)$, $\texttt{after-pop}(ql)$, $\texttt{before-push}(qr)$, $\texttt{after-pop}(qr)$ and $\texttt{before-pop}(pq)$ to $q$.
To this point, $\eta$ maps $\Oh(\frac{s(n)}{\log n})$ variables to each node of $T$.
Create a tree-depth decomposition of $G[R]$ as follows: start with $T$, end replace each node $q$ of $T$ with a path consisting of variables $\eta^{-1}(q)$, ordered arbitrarily. These paths are organized in a tree in the same way as the original nodes: the last vertex of a path corresponding to a node $q$ becomes the parent of the first node of each path corresponding to a child of $q$. 
Observe that whenever two variables of $R$ appear in the same clause, then the nodes to which their blocks are mapped by $\eta$ are either equal, or they are in the ancestor-descendant relation. Therefore, it is easy to see that indeed we have constructed a valid tree-depth decomposition of $G[R]$. Since $T$ has depth $\Oh(\log n)$ and the pre-image of each node of $T$ under $\eta$ has size $\Oh(\frac{s(n)}{\log n})$, the depth of the decomposition is $\Oh(s(n))$.

It remains to consider the variables and clauses created in computation gadgets for each node $q$.
By Lemma~\ref{lem:compGadget}, every such gadget has a tree-depth decomposition of depth $\Oh(\frac{s(n)}{\log n}\cdot \log n + h) = \Oh(s(n))$, where $h$ is the current stack depth.
Observe that, among the variables of $R$, the clauses of each gadget connect only variables from blocks mapped to $q$ by $\eta$, and variables from the $\texttt{stack}$ blocks corresponding to the edges on the path from $q$ to the root of $T$. Hence, we can take the tree-depth decomposition of the gadget of depth $\Oh(s(n))$, remove from it all the variables contained in $R$, and attach the resulting decomposition as a new subtree below the deepest vertex of $\eta^{-1}(q)$. By performing this operation for each computation gadget, we obtain a tree-depth decomposition of the whole prime graph $G$ that has depth $\Oh(s(n)+s(n)) = \Oh(s(n))$.
It is straightforward to verify, using Lemma~\ref{lem:compGadget}, that all the described constructions can be performed in logspace.
Therefore $L \lred \text{td-CNF-SAT}[c \cdot s(n)]$ for some constant $c$ depending on the machine $M$; since we assume $s(n)$ to be a nice function, $\text{td-CNF-SAT}[c \cdot s] \lred \text{td-CNF-SAT}[s]$.
\end{proof}

Lemmas~\ref{lem:tdRegularize} and~\ref{lem:tdHardness} show that $\text{td-CNF-SAT}[s]$ is hard for $\NAuxSA{\poly}{s/\log}{s}$ under logspace reductions, and
by Lemma~\ref{lem:equivalenceOfProblems}, so is $\tdEx{s}$.
Lemmas~\ref{lem:tdContain}, \ref{lem:tdRegularize} and the fact that the closure of $\NAuxSAplain{\poly}{\log}{s}$ under logspace reductions is $\NAuxSAplain{\poly}{\log}{s(\poly)}$,
give the following chain of containments ($[A]^L$ is the class of problems reducible to $A$ in logspace):
$$ [\tdEx{s}]^L \subseteq \NAuxSA{\poly}{\log}{s(\poly)} \subseteq \NAuxSA{\poly}{s(\poly)/\log}{s(\poly)} \subseteq$$
$$\subseteq \text{reg-}\NAuxSA{\poly}{s(\poly)/\log}{s(\poly)} \subseteq [\tdEx{s(\poly)}]^L \subseteq [\tdEx{s}]^L$$
Therefore, all containments must be equalities, which concludes the proof of Theorems~\ref{thm:main-td} and \ref{thm:space-increase}.
%
%
%
Note that for an unbounded (or polynomial) stack, Theorem~\ref{thm:space-increase} implies that space is unbounded too; in other words, \NL (that is, non-deterministic logspace) machines augmented with an unbounded auxiliary stack have the same power as \cfont{NP}, an observation already made by Vinay and Chandru~\cite{VinayC90}.

By Theorem~\ref{thm:main-td}, to prove a determinization result for $\NAuxSAplain{\poly}{\log}{s(\poly)}$ we only need such a result for \tdEx{s}.

\begin{lemma}\label{lem:deter-3col}
$\tdEx{s}$ can be solved in time $3^s\cdot \poly(n)$ and space $\Oh(s+\log n)$.
\end{lemma}
\begin{proof}
Let $G$ be the input instance of \tdEx{s}, and let $(\mathcal{T},\mu)$ be the given tree-depth decomposition of $G$. By abuse of notation, we identify the vertices of $G$ with their images under $\mu$. For every $u\in V(G)$, let $\tail(u)$ be the set of vertices on the path from $u$ to the root of its tree in $\mathcal{T}$, excluding $u$, and let $\tree[u]$ be the set of vertices contained in the subtree of $\mathcal{T}$ rooted at $u$, including $u$. For every $u\in V(G)$ and every proper coloring $\phi$ of $G[\tail(u)]$ into $3$ colors, let $f(u,\phi)$ be the Boolean value denoting whether $\phi$ can be extended to a proper $3$-coloring of $G[\tree[u]\cup \tail(u)]$. Then clearly $f(u,\phi)$ is true if and only if it is possible to extend $\phi$ to $\phi'$ by assigning $u$ one of the three colors in such a manner that $\phi'$ remains a proper $3$-coloring of $G[\tail(u)\cup \{u\}]$, and $f(v,\phi')$ is true for every child $v$ of $u$. Whether $G$ has a proper $3$-coloring is equivalent to the conjunction of values $f(r,\emptyset)$ over the roots $r$ of trees in $\mathcal{T}$.

We give a recursive procedure for computing values $f(u,\phi)$; the whole problem then reduces to computing $f(r,\emptyset)$ for every root $r$ of a tree in $\mathcal{T}$. This recursive procedure simply browses through all three possible extensions $\phi'$ of $\phi$ to a proper $3$-coloring of $G[\tail(u)\cup \{u\}]$, and calls itself recursively to compute $f(v,\phi')$ for all children $v$ of $u$; in particular, no memoization of computed values is done. Note that the recursion tree stops in the leaves of $\mathcal{T}$. The correctness of the algorithm follows directly from the discussion of the previous paragraph. As far as the space usage is concerned, at each point the algorithm maintains identifier of the current vertex $u$, of logarithmic length, and a stack of $\Oh(s)$ calls to the procedure computing $f(\cdot,\cdot)$. The data stored for each call requires constant space; note that there is no need to memorize the identifier of the vertex, because it can be recomputed when returning from a subcall. Hence, the space complexity of the algorithm is $\Oh(s+\log n)$. To analyze the running time, observe that for each pair $(u,\phi)$, where $\phi$ is a proper $3$-coloring of $G[\tail(u)]$, throughout the whole algorithm there will be at most one call to $f(u,\phi)$; this is because whenever recursing, we are considering an extension of the current coloring. Thus, the whole recursion tree will have at most $n\cdot 3^s$ nodes. Since the computations at each node are done in polynomial time, it follows that the running time of the algorithm is $3^s\cdot \poly(n)$.
\end{proof}

Therefore, for any nice $s(n)\geq \log n$, we have $\tdEx{s}\in \DSpace{s}$. Theorem~\ref{thm:main-td}, Lemma~\ref{lem:deter-3col}, and the observation that $\DSpace{s(\poly)}$ is closed under logspace reductions, yield Theorem~\ref{thm:determinization}.


\subsection{Characterization via alternating machines} 
An \emph{alternating Turing machine} (ATM) is a Turing machine with a partition of states into \emph{existential} or \emph{universal} states. 
For an ATM and an input word, an \emph{accepting tree} is a finite tree $T$ labeled with machine configurations, such that: the root is labeled with the initial configuration, every node with a configuration in an existential state in $T$ has one child labeled with a next configuration (one reachable in one step according to the machine's transition rules), every node with a configuration in a universal state in $T$ has all possible next configurations as children, and all leaves are accepting configurations.
An ATM accepts an input word in time $t$, space $s$, and \emph{treesize} $z$, if there is an accepting tree with root-to-leaf distances at most $t$, configurations using at most $s$ space, and at most $z$ tree nodes.

Similarly as ATMs proved to be a useful computational model, giving a new, unified view on various complexity issues, the notion of treesize introduced by Ruzzo~\cite{Ruzzo80} allowed to see various classes under a common light and simplify a few containment proofs.
In particular, Ruzzo showed that $\NAuxPDA{\poly}{s} = \ASpSz{s}{\poly}$.
We show that bounding the time (as opposed to space) of a polynomial treesize ATM, leads to the classes corresponding to small tree-depth, as opposed to small treewidth.

\begin{theorem}\label{thm:alternation}
Let $s(n)\geq \log^2(n)$ be a nice function. Then
$$\NAuxSA{\poly}{\log}{s(\poly)} = \ATiSz{s(\poly)}{\poly}.$$
\end{theorem}
\begin{proof}
For one containment $(\subseteq)$ we show that $\tdEx{s}$, which is hard for the former class by Theorem~\ref{thm:main-td}, is contained in the latter.
Indeed, a straightforward algorithm traverses a given tree-depth decomposition of the input graph top-down by existentially guessing a color of each encountered vertex, checking its compatibility with previous guesses, and universally guessing which subtree to proceed into.
The bounds on time and treesize of an alternating Turing machine executing this algorithm follow directly from the bounds on the depth ($s$) and size (polynomial) of the decomposition's tree.

For the other containment, we simulate an alternating Turing machine with a NAuxSA machine.
This is done exactly as in Ruzzo's simulation with NAuxPDA machines (Theorem 1 in~\cite{Ruzzo80}), 
except that all the configurations on the current path are remembered on the stack by only writing the difference ({\em{log of changes}}) from the previous configuration: the constant-size description consists of the new state, the direction of movement for each head and the symbol written on the worktape.
Observe that in this manner, the machine can within logarithmic working space retrieve all the information needed to verify availability of a transition:
\begin{itemize}
\item The current simulated state is on the top of the stack;
\item The current head positions can be recalculating by browsing through the stack and applying consecutive moves;
\item The symbol under the worktape head can be recalculated by finding  on the stack the latest symbol written on its current position.
\end{itemize}
Consequently, deterministic steps of the machine can be simulated by verifying the applicability of the transition, and pushing onto the stack the constant-size log of changes leading to the new configuration. Similarly, existential guesses are simulated with the machine's own non-determinism.
Universal guesses are simulated by choosing one possibility and then backtracking (popping the stack) to choose the next possibility, as described by Ruzzo.
Such backtracking corresponds exactly to traversing an accepting tree of the ATM, hence the running time is polynomial.
At most $s(\poly(n))$ steps are described on the stack at any time, each requiring a constant number of symbols due to keeping only the log of changes.
\end{proof}

\section{{\normalfont\scshape{Dominating Set}} on graphs of small treedepth}\label{sec:domset}
In this section we show how to solve \textsc{Dominating Set}, or even count the number of dominating sets of any cardinality, in time $3^s \cdot \poly(n)$ and space $\Oh(s \cdot \log n)$, 
given a tree-depth decomposition of depth $s$.
Recall that for a graph $G$, a set $S\subseteq V(G)$ is a \emph{dominating set} if every vertex of $G$ either is in $S$, or is adjacent to some vertex of $S$.
We first describe an algorithm working in space $\poly(n)$.

\begin{lemma}\label{lem:ds}
There exists an algorithm that, given a graph $G$ on $n$ vertices and its tree-depth decomposition of depth $s$, runs in time $3^s\cdot \poly(n)$ and space $\poly(n)$, and outputs a sequence $(q_i)_{0\leq i\leq n}$, where $q_i$ is the number of dominating sets of $G$ of cardinality $i$.
\end{lemma}
\begin{proof}
We will work in the ring of polynomials $\mathbb{Z}[x]$, where $x$ is a formal variable. The algorithm will compute polynomial $P(x)=\sum_{i=0}^n q_i x^i$, whose coefficients constitute the output.

Let $(\mathcal{T},\mu)$ be the given tree-depth decomposition of $G$; recall that $\mathcal{T}$ is a rooted forest of depth at most $s$. By abuse of notation, we identify the vertices of $G$ with their images under $\mu$. Let us introduce some notation relating to $\mathcal{T}$. For $u\in V(G)$, by $\tail[u]$ we denote the set of vertices on the path in $\mathcal{T}$ from $u$ to the root of its tree in $\mathcal{T}$. By $\tree[u]$ we denote the set of all the vertices contained in the subtree of $\mathcal{T}$ rooted at $u$, including $u$. Define $\tail(u)=\tail[u]\setminus \{u\}$ and $\tree(u)=\tree[u]\setminus \{u\}$. By $\chld(u)$ we denote the set of children of $u$ in $\mathcal{T}$. For a function $h$, an argument $e$ outside the domain of $h$, and a value $\alpha$, by $h[e\to \alpha]$ we denote the function $h$ extended by adding $e$ to the domain and mapping it to $\alpha$.

Let $\Sigma=\{\Aa,\Ff,\Tt\}$, where $\Aa$, $\Ff$, $\Tt$ are some symbols (the reader should think of them as Allowed, Forbidden, and Taken, respectively). For any vertex $u$ of $G$ and function $\phi\colon \tail(u)\to \Sigma$, define $f(u,\phi)\in \mathbb{Z}[x]$ as $\sum_{i=0}^n a_i x^i$, where $a_i$ is the number of $i$-element subsets $X\subseteq \tree[u]$ such that $X\cup \phi^{-1}(\Tt)$ dominates $\tree[u]$ and no vertex of $\phi^{-1}(\Ff)$ in $G$. 
Similarly, for a function $\psi\colon \tail[u]\to \Sigma$, define $g(u,\psi)\in \mathbb{Z}[x]$ as $\sum_{i=0}^n b_i x^i$, where $b_i$ is the number of $i$-element subsets $X\subseteq \tree(u)$ such that $X\cup \psi^{-1}(\Tt)$ dominates $\tree(u)$ and no vertex of $\phi^{-1}(\Ff)$ in $G$. Note that if there are two adjacent vertices $v,v'\in \tail(u)$ with $\phi(v)=\Tt$ and $\phi(v')=\Ff$, then no set $X$ can satisfy the requirements above and hence $f(u,\phi)=0$. Similarly for $g$ and $\psi$.

We remark that values $f$ and $g$ are exactly what one would obtain by applying the M\"obius transform to the standard definition of states for dynamic programming for {\textsc{Dominating Set}} (that is, we count sets that dominate any subset of $\phi^{-1}(\{\Aa,\Tt\})$ in the tail, instead of exactly specifying which vertices are to be dominated). This transform translates subset convolutions used in the standard dynamic programming to pointwise products, which is the crucial idea behind the proof. Since the algorithm is not complicated, we prefer to present it directly after applying the transform.

We now give recursive equations on the values of $f(\cdot,\cdot)$ and $g(\cdot,\cdot)$. First, observe that for each $v\in \chld(u)$ we have $\tail(v)=\tail[u]$. Then, it is easy to verify that the following equation holds for each $u\in V(G)$ with $\chld(u)\neq \emptyset$ and each $\psi\colon \tail[u]\to \Sigma$:
\begin{equation}\label{e:g}
g(u,\psi)=\prod_{v\in \chld(u)} f(v,\psi).
\end{equation}
Indeed, every set $X\subseteq \tree(u)$ that contributes to some coefficient of $g(u,\psi)$ can be partitioned into $\{X\cap \tree[v]\, \colon\, v\in \chld(u)\}$. Each set $X\cap \tree[v]$ contributes to the coefficient by $x^{|X\cap \tree[v]|}$ of $f(v,\psi)$, and hence when computing the product the formal variable $x$ correctly keeps track of the cardinality. When $\chld(u)=\emptyset$, then $\tree(u)=\emptyset$ and we can compute $g(u,\psi)$ directly from the definition:
\begin{equation}\label{e:base}
g(u,\psi)=\begin{cases}1\qquad \textrm{if there is no edge between $\psi^{-1}(\Tt)$ and $\psi^{-1}(\Ff)$,}\\ 0\qquad \textrm{otherwise.}\end{cases}
\end{equation}

We now proceed to setting up the equation for $f(\cdot,\cdot)$. Take any $u\in V(G)$ and $\phi\colon \tail(u)\to \Sigma$. Then, it is easy to verify that the following equation holds:
\begin{equation}\label{e:f}
f(u,\phi)=g(u,\phi[u\to \Aa])-g(u,\phi[u\to \Ff])+x\cdot g(u,\phi[u\to \Tt]).
\end{equation}
The term $x\cdot g(u,\phi[u\to \Tt])$ counts the contribution from sets $X$ that contain $u$. The term $g(u,\phi[u\to \Aa])$ counts the contribution from all sets $X$ that do not contain $u$, regardless of whether they dominate $u$ or not, whereas by subtracting the term $g(u,\phi[u\to \Ff])$ we remove the contribution from sets that do not contain or dominate $u$. Observe that if $u$ has a neighbor in $\phi^{-1}(\Tt)$, i.e., it is already dominated by $\phi^{-1}(\Tt)$, then the subtracted term $g(u,\phi[u\to \Ff])$ will be a zero polynomial. This corresponds to the fact that in this case we do not need to care about domination of $u$ by $X$. Similarly, if $u$ has a neighbor in $\phi^{-1}(\Ff)$, then $g(u,\phi[u\to \Tt])$ will be a zero polynomial. This corresponds to the fact that in this case it is not allowed to take $u$ to $X$.

Finally, observe that
\begin{equation}\label{e:whole}
P=\prod_{u\in \roots} f(u,\emptyset),
\end{equation}
where $\roots$ is the set of roots of the trees in forest $\mathcal{T}$.

We now give the algorithm that computes $P$. The algorithm uses two mutually recursive functions that compute the values of $f(\cdot,\cdot)$ and $g(\cdot,\cdot)$, respectively. The polynomial $P$ is computed using equation~\eqref{e:whole} by a sequence of calls to the procedure computing $f(\cdot,\cdot)$. The procedure computing $f(\cdot,\cdot)$ applies equation~\eqref{e:f} and calls $g(\cdot,\cdot)$ recursively. Similarly, the procedure computing $g(\cdot,\cdot)$ applies equation~\eqref{e:g} and calls $f(\cdot,\cdot)$ recursively, or uses the base case~\eqref{e:base}. In particular, no memoization of computed values is performed.

The correctness of the algorithm follows from equations~\eqref{e:g}--\eqref{e:whole}. Note that at each moment, the space used by the algorithm is composed of a stack of at most $2s+1$ recursive calls to $f(\cdot,\cdot)$, $g(\cdot,\cdot)$, and the main procedure computing $P$, and for each of these calls we can store just the partial result of computation being one polynomial from $\mathbb{Z}[x]$ (in case of equations~\eqref{e:g} and~\eqref{e:whole}, this will be the product calculated for a prefix). These polynomials have degrees bounded by $n$ and their coefficients have values between $0$ and $2^n$, hence the total space usage of the algorithm is $\Oh(sn^2)$. Finally, to estimate the running time observe that for every pair $(u,\phi)$, where $u\in V(G)$ and $\phi\colon \tail(u)\to \Sigma$, throughout the whole computation there will be at most one call to computing $f(u,\phi)$. This is because when recursing, the new function $\phi$ is always an extension of the previous one. Similarly, for every pair $(u,\psi)$, where $u\in V(G)$ and $\psi\colon \tail[u]\to \Sigma$, throughout the whole computation there will be at most one call to computing $g(u,\psi)$. The total number of such pairs $(u,\phi)$ and $(u,\psi)$ is at most $2n\cdot 3^s$, which implies that the whole recursion tree has at most this many nodes. Since the work done at each node is polynomial in $n$, we conclude that the algorithm runs in time $3^s\cdot \poly(n)$.
\end{proof}

\newcommand{\F}{\mathbb{F}}

We now show how to improve the space usage of the above algorithm to $\Oh(s \cdot \log n)$.
Given a graph $G$ on $n$ vertices and its tree-depth decomposition of depth $s$, consider again the polynomial $P(x)=\sum_{i=0}^n q_i x^i$ over a single variable $x$, 
with integer coefficients $q_i$ equal to the number of dominating sets of $G$ of cardinality exactly $i$ (in particular, $0 \leq q_i \leq 2^n$).
We use the fact that given a prime number $p\leq 2n+2$ and some element $a$ of the Galois field $\F_p$, the value of $(P(a) \bmod p)$ can be computed in time $\Oh(3^s \cdot \poly(n))$ and space $\Oh(s \cdot \log n)$.
Indeed, one can recursively compute all $\Oh(3^s \cdot n)$ values of $(f(\cdot,\cdot)(a) \bmod p)$ and $(g(\cdot,\cdot)(a) \bmod p)$, as described in the proof of Lemma~\ref{lem:ds}.
On each of the $\Oh(s)$ recursion levels, we need to maintain only one number in $\F_p$ (describing a partial sum or product), and each value can be added or multiplied 
(to the partial sum or product that requires it) in time $\Oh(\poly(n))$ and space $\Oh(\log n)$.


With such a procedure in hand, the following theorem describes how to recover the exact coefficients of $P(x)$ using interpolation. 
This is done using the Chinese remainder theorem, and applying a number-theoretic transform, that is, the discrete Fourier transform specialized to the field of integers mod $p$, 
for prime~$p$. Effectively, this technique boils down to evaluating the polynomial in many points and computing a weighted sum of the results.
We need the following simple corollary of the prime number theorem. In fact, it can be proved for $n_0 = 21$ using explicit bounds given by~\cite[Theorem 4]{rosser1962} and hand computation for small enough $n$.

\begin{fact}\label{fact:prime}
	There is an $n_0 \in \mathbb{N}$ such that 
	for all $n\geq n_0$, the product of primes strictly between $n$ and $2n$ is larger than $2^n$.
\end{fact}	

\begin{theorem}
	Let $P(x)=\sum_{i=0}^n q_i x^i$ be a polynomial over one variable $x$, of degree at most $n$ and with integer coefficients satisfying $0 \leq q_i \leq 2^n$, for $i=0,\dots,n$.
	Suppose that given a prime number $p\leq 2n+2$ and $a\in \F_p$, the value of $(P(a) \bmod p)$ can be computed in $T$ time and $S$ space.
	Then given $k\in\{0,\dots,n\}$, the value $q_k$ can be computed in $\Oh(T\cdot \poly(n))$ time and $\Oh(S + \log n)$ space. 
\end{theorem}
\begin{proof}
	We first show how to compute $q_k \bmod p$, given $k$ and a prime $p$ with $n+1 < p < 2n+2$.
	Let $\alpha$ be a primitive element of the field $\F_p$, that is, a generator of the multiplicative group $\F_p^*$ (in other words,  $\F_p^* = \{\alpha^0, \alpha^1, \dots, \alpha^{p-2}\}$).
	Such an element can be found in polynomial time by trying all elements of $\F_p^*$ and testing whether $\alpha^1,\dots,\alpha^{p-2} \neq 1$.
	Compute in $\F_p$ the value
	$$q'_k = -\sum_{i=0}^{p-2} P(\alpha^i) \cdot \alpha^{-ik}.$$
	This takes $\Oh(T \cdot n)$ time and $\Oh(S + \log n)$ space.
	
	We claim $q'_k = q_k \bmod p$.
	To show this, first notice that for $a\in \F_p \setminus \{0\}$, we have
	$$(a-1) \cdot \sum_{i=0}^{p-2} a^i\ =\ a^{p-1} - 1\ =\ 0 \pmod{p}.$$
	Hence, if $a\neq 1$ then $\sum_{i=0}^{p-2} a^i = 0 \pmod{p}$, while for $a=1$, we have
	$\sum_{i=0}^{p-2} a^i = p-1 = -1 \pmod{p}$.
	Thus, for any $k\in\Z$ we have the following:
	$$\sum_{i=0}^{p-2} (\alpha^k)^i\ =\ [k \equiv 0\!\mod (p-1)] \cdot (-1) \pmod{p}.$$
	Here, $[ \cdot ]$ denotes Iverson's notation: the value is $0$ or $1$ depending whether the predicate in the brackets is false or true, respectively.
	Since $n < p-1$, for integers $j,k\in \{0,\dots,n\}$ we have $j-k \equiv 0\!\mod (p-1)$ if~and~only~if $j=k$.
	Therefore, as claimed, we conclude that
	\begin{align*}
	q'_k &=\quad -\sum_{i=0}^{p-2}\left( \sum_{j=0}^n q_j\cdot (\alpha^i)^j\right) \cdot \alpha^{-ik} \quad=\quad
	     -\sum_{j=0}^n q_j \sum_{i=0}^{p-2} (\alpha^{j-k})^i \quad=\\
	     &=\quad -\sum_{j=0}^n q_j \cdot [j=k] \cdot (-1) \quad=\quad
	     q_k\pmod{p}.
	\end{align*}
	
	Therefore, for every integer $p$ strictly between $n+1$ and $2(n+1)$, we can check whether it is prime (by brute-force) and compute $q_k \bmod p$ in $\Oh(T \cdot n)$ time and $\Oh(S + \log n)$ space.
	From Fact~\ref{fact:prime} and the Chinese remainder theorem, it follows that $q_k$ is non-zero if and only if $(q_k\bmod p)$ turns out to be non-zero for at least one $p$ (we assume $n+1\geq n_0$, as otherwise $q_k$ can be computed by brute-force).
	Moreover, given a list of pairs $(p, q_k\bmod p)$ for the primes $p$ between $n+1$ and $2(n+1)$,
	the exact value of $q_k$ can be recovered with a logspace (and hence polynomial time) algorithm by Chiu et al.~\cite[Theorem 3.3]{ChiuDL01}, which is an effective version of the Chinese remainder theorem.
	Note that we do not need to simultaneously store all values $q_k \bmod p$ for different $p$: we use the compositionality of logspace algorithms instead, that is, values output by our algorithm are recomputed on the fly as needed by the remaindering algorithm, multiplying the running times and adding the space bounds of the two algorithms.
\end{proof}

\begin{corollary}
	There exists an algorithm that, given a graph $G$ on $n$ vertices and its tree-depth decomposition of depth $s$, runs in time $3^s\cdot \poly(n)$ and space $\Oh(s \cdot \log n)$, and outputs a sequence $(q_i)_{0\leq i\leq n}$, where $q_i$ is the number of dominating sets of $G$ of cardinality $i$.
\end{corollary}

\section{Conclusions}\label{sec:conc}
In what follows we assume for conciseness that $s(n)$ is a nice function satisfying $s(n^c)=\Oh(s(n))$ for each constant $c$ (this includes $s(n)=\lg^k n$ for $k\geq 1$, in particular).

The hierarchy of graph parameters of Corollary~\ref{cor:td_pw_tw} together with Theorems~\ref{thm:pwCompleteness},~\ref{thm:main-td}, and~\ref{thm:alternation} implies the following hierarchy of complexity classes between $\cfont{NL}$ and $\cfont{NP}$.
\medskip

\newcommand{\balphantom}{\rule[-18pt]{0pt}{0pt}}
\begin{tabular}{c@{ }c@{ }c@{ }c@{ }c@{\ }c@{\ }c}
$\NAuxSA{\poly}{\log}{s}$          &=& $[\tdEx{s}]^L$             &=&$\ATiSz{s}{\poly}$ &$\subseteq$&$\DSpace{s}$\\
\balphantom                        & &\rotatebox{-90}{$\subseteq$}& & \\  
$\NTiSp{\poly}{s}$                 &=& $[\pwEx{s}]^L$             &=&$\NTiSp{\poly}{s}$ \\
\balphantom                        & &\rotatebox{-90}{$\subseteq$}& & \\
$\NAuxPDA{\poly}{s}$               &=& $[\twEx{s}]^L$             &=&$\ASpSz{s}{\poly}$&$\subseteq$&$\DTime{2^{\Oh(s)}}$\\
\balphantom                        & &\rotatebox{-90}{$\subseteq$}& & \\
$\NAuxSA{\poly}{\log}{s\cdot \log}$&=& $[\tdEx{s\cdot \log}]^L$    &=&$\ATiSz{s \cdot \log}{\poly}$&$\subseteq$&$\DSpace{s\cdot\log}$\\
\end{tabular}
\bigskip

In particular, when considering functions $s(n) = \log^k(n)$, the classes have sometimes been considered under different names:
\begin{itemize}
\item $\NAuxSA{\poly}{\log}{\log^{k}}$ was named \cfont{DC$^{k-1}$} (for \emph{divide and conquer}) in~\cite{AkatovG10,akatov2010exploiting},
\item $\NTiSp{\poly}{\log^k}$ are known as $\NSC{k}$ (the non-deterministic variant of \emph{Steve's Class}),
\item $\NAuxPDA{\poly}{\log^k}$ is shown equal to a class named $\SACq{k}$ in~\cite{AllenderCLPT14}.
\end{itemize}
This yields the following hierarchy:
$$
\cfont{L} \subseteq
\begin{matrix} \NL\\\parallel\\\NSC{1} \end{matrix} \subseteq
\begin{matrix} \cfont{SAC}^1\\\parallel\\\SACq{1} \end{matrix} \subseteq
\cfont{DC}^1 \subseteq
\dots \subseteq
\cfont{DC}^{k-1} \subseteq
\NSC{k} \subseteq
\SACq{k} \subseteq
\cfont{DC}^k \subseteq
\dots \subseteq
\cfont{NP} $$

We conclude with an open question stemming from this work. 
In Section~\ref{sec:treedepth} we have shown that {\textsc{3Coloring}} is complete for $\NAuxSA{\poly}{\log}{s}$ when a tree-depth decomposition of depth $s(n)$ is given on the input. 
By Lemma~\ref{lem:equivalenceOfProblems}, the same holds for equivalent CSP-like problems, 
like {\textsc{CNF-SAT}} (with primal or incidence graph), whereas {\textsc{Independent Set}} and {\textsc{Dominating Set}} are hard for the same class.

It is not hard to see that {\textsc{Independent Set}} actually can be solved in the complexity class $\NAuxSA{\poly}{\log}{s}$, using an approach very similar to that of Lemma~\ref{lem:deter-3col} as follows. 
The algorithm traverses the treedepth decomposition in the prefix order, nondeterministically guessing a maximum-size independent set $X$ on the fly, 
and storing the following information: on the stack we store the intersection of $X$ with the path from the current vertex to a root of the decomposition, 
whereas in the working memory we store the number of vertices from $X$ found so far.
By Theorem~\ref{thm:determinization}, this means that {\textsc{Independent Set}} on a graph given with tree-depth decomposition of depth $s$ can be solved deterministically in space  $\Oh(s+\log n)$.

As far as {\textsc{Dominating Set}} is concerned, in Section~\ref{sec:domset} we demonstrated how using the algebraic approach of Lokshtanov and Nederlof~\cite{LokshtanovN10}, and of F\"urer and Yu~\cite{FurerY14}, 
one can obtain an algorithm for {\textsc{Dominating Set}} with running time $3^s\cdot \poly(n)$ and space complexity $\Oh(s \cdot \log n)$. 
Thus, it is unclear to us whether the problem {\textsc{Dominating Set}} on graphs with treedepth decompositions of width $s(n)$ belongs to $\NAuxSA{\poly}{\log}{s}$. 
Observe that if this would be the case, then by Theorem~\ref{thm:determinization} it should be solvable in space $\Oh(s+\log n)$; however, already achieving space complexity $\Oh(s \cdot \log n)$ was highly nontrivial.

\subparagraph*{Acknowledgements.} The authors thank Yoichi Iwata for pointing out that {\textsc{Independent Set}} on graphs given 
with a treedepth decomposition of width $s(n)$ is actually in the class $\NAuxSA{\poly}{\log}{s}$.

\printbibliography
\appendix
\section{Reductions preserving structural parameters}\label{app:preserve}
To capture the structural dependencies in reductions for a more uniform proof, we use the following definition borrowed from Chen and M\"{u}ller~\cite{ChenM14}.
While very similar to a tree decomposition, it is not limited to trees and allows an edge to be covered by two adjacent bags instead of one, which turns out to give a generalization with better properties. 

\begin{definition}\label{def:deconstruction}
For graphs $G,H$, an \emph{$H$-deconstruction of $G$} is a family $(B_h)_{h\in V(H)}$ of subsets of $V(G)$ (called \emph{bags}) such that every vertex of $G$ is in some bag, every edge of $G$ has both endpoints contained in one, or two adjacent (in $H$) bags, and for each vertex $v\in V(G)$ the subset $\{h \in V(H) \mid v\in B_h\}$ is connected in $H$.
The \emph{width} of a deconstruction is the maximum size of a bag or a union of two adjacent bags.
\end{definition}

We observe that in many reductions, the output graph can be deconstructed into the input graph (or e.g. the incidence graph of the input formula) with constant width.
We first show that this guarantees the reduction preserves structural parameters.

\begin{lemma}\label{lem:deconstruction}
Let $\pi\in\{\td,\pw,\tw\}$. There is a logspace algorithm that given graphs $G,H$, an $H$-deconstruction of $G$ of width $w$, and a $\pi$-decomposition of $H$ of width/depth $w_h$, outputs a $\pi$-decomposition of $G$ of width/depth at most $w \cdot (w_h+1)$.
In particular, $\pi(G) \leq w \cdot \pi(H)$.
\end{lemma}
\begin{proof}
Let $(B_h)_{h\in V(H)}$ be an $H$-deconstruction of $G$ of width $w$.

For treewidth and pathwidth, observe that if $(\mathcal{T},(C_t)_{t\in \mathcal{T}})$ is a tree (or path) decomposition of $H$ of width $w_h$, then the same tree $\mathcal{T}$ with bags defined as $C'_t = \bigcup_{h\in C_t} B_h$ is a valid decomposition of $G$ of width at most $w \cdot (w_h+1)$.

For tree-depth, let $(\mathcal{T},\mu)$ define a tree-depth decomposition of $H$ of depth $w_h$.
Create sets $M_t$ for nodes $t$ of $\mathcal{T}$ and place each vertex $v$ of $G$ in $M_t$ where $t$ is the lowest common ancestor of $\{\mu(h) \mid v \in B_h\}$ in $\mathcal{T}$.

Observe that by definition of a deconstruction, if $v \in B_h$ and $v \in B_{h'}$ for some $h,h'\in V(H)$, then there is a path connecting $h$ and $h'$ in $H$ containing only vertices $h''$ such that $v\in B_{h''}$;
hence there is an $h'' \in V(H)$ such that $\mu(h'')$ is a common ancestor of $\mu(h)$ and $\mu(h')$ in $\mathcal{T}$ and $v \in B_{h''}$.
Therefore, if $h_0$ is the lowest common ancestor of $\{\mu(h) \mid v \in B_h\}$ in $\mathcal{T}$, then also $v\in B_{h_0}$, and hence every vertex $v$ is put into a set $M_t$ such that $v \in B_{\mu^{-1}(t)}$. That is, $M_t \subseteq B_{\mu^{-1}(t)}$, which implies $|M_t|\leq w$ for all $t\in V(\mathcal{T})$.

Let us then modify $\mathcal{T}$ by replacing every node $t\in V(\mathcal{T})$ by a path of $|M_t|$ nodes, and define a bijection $\mu'$ between vertices of $G$ and nodes of this tree that maps vertices in $M_t$ to nodes of the path that replaced $t$ (in any order). Let $\mathcal{T}'$ be the modified decomposition.
Clearly $\mathcal{T}'$ has depth at most $w\cdot w_h$.
To check that $\mathcal{T}'$ with $\mu'$ defines a valid tree-depth decomposition, consider any edge $uv$ of $G$.
By definition of a deconstruction there are adjacent or equal vertices $h,h'$ in $H$ such that $u\in B_h$ and $v\in B_{h'}$.
Since $\mu(h)$ is an ancestor of $\mu(h')$ or vice versa, $u$ was assigned to a set $M_t$ such that $t$ is an ancestor of $\mu(h)$ and $v$ was assigned to a set $M_{t'}$ such that $t'$ is an ancestor of $\mu(h')$, it must be that $t$ is an ancestor of $t'$ or vice versa.
Hence $\mu'(u)\in M_t$ is an ancestor of $\mu'(v)\in M_{t'}$ or vice versa.
\end{proof}

Many reductions between NP-complete graph problems introduce components of bounded size replacing every edge of the original graph, or more generally, attach small components to cliques of the original graph.
We need the following lemma to show that such reductions also preserve structural parameters.
For a graph $G$ and a vertex set $S\subseteq V(G)$, the \emph{subgraph of $G$ induced by $S$}, denoted $G[S]$, is the graph with vertex set $S$ and edge set $E(G) \cap (S\times S)$.
We write $G-S$ for $G[V(G)\setminus S]$.

\begin{lemma}\label{lem:incidenceBound} 
Let $\pi\in\{\td,\pw,\tw\}$. Let $G$ be an induced subgraph of $G'$ such that for each connected component $C$ of $G'-V(G)$ we have that $C$ has at most $c$ vertices and the neighborhood of $C$ in $V(G)$ is a clique in $G$, for some constant $c$. Then
$$\pi(G') \leq \pi(G)+c.$$
Furthermore, given $G,G'$ and any $\pi$-decomposition of $G$, one can compute a $\pi$-decomposition of $G'$ of width/depth larger by at most $c$ in logspace.
\end{lemma}
\begin{proof}
For \td, in any tree-depth decomposition, the vertices of a clique in $G$ must be mapped to nodes fully ordered by the ancestor relation.
We may thus simply take each connected component $C$ of $G'-V(G)$, examine the placement of the clique $N(C)\subseteq V(G)$ in the given tree-depth decomposition of $G$, and attach the vertices of $C$ as a path of length $|C|$ below the lowest node the clique $N(C)$ maps to. In this manner we create a tree-depth decomposition of $G'$ of depth larger than the input tree-depth decomposition of $G$ by at most $c$.

For \tw and \pw, we use the fact that each clique in $G$ must be fully contained in some bag of a decomposition (e.g., \cite[Lemma 1]{Bodlaender05}).
Hence, for each connected component $C$ of $G'-V(G)$, we find a bag that contains all the vertices of the clique $N(C)\subseteq V(G)$, and create a copy of this bag into which all the vertices of $C$ are added. It is straightforward to arrange the new bags in the decomposition. In this manner we construct a decomposition of $G'$ of width larger than the original decomposition of $G$ by at most $c$.

For all three parameters, it is trivial to implement the described procedure in logspace.
\end{proof}

We are now ready to show how standard reductions for some example NP-complete problems prove them to be equivalent to \exampleProblem, or at least as hard, in our setting.

{\renewcommand{\thetheorem}{\ref{lem:equivalenceOfProblems} (restated)}
\begin{lemma}
The following problems are equivalent under logspace reductions that preserve structural parameters: \textsc{3Coloring}, CNF-SAT (using a decomposition of the primal graph), $k$-SAT (using a decomposition of either the primal or incidence graph) for each $k\geq 3$.

Furthermore, the following problems admit logspace reductions that preserve structural parameters from the above problems: \textsc{Vertex Cover}, \textsc{Independent Set}, \textsc{Dominating Set}.
\end{lemma}
\addtocounter{theorem}{-1}}
\begin{proof}
All the following reductions are standard, and hence we keep the description concise. Also, it will be straightforward to verify that they can be implemented in logspace. The only non-trivial check will be to verify, using Lemmas~\ref{lem:incidenceBound}  and~\ref{lem:deconstruction}, that the structural parameters are preserved.

\medskip

\noindent CNF-SAT (primal graph) $\lred$ $k$-SAT (primal graph) (for any $k\geq 3$):\\[0.1cm]
Replace every clause $(l_1 \vee l_2 \vee \dots \vee l_\ell)$ of length $\ell$ with clauses $(l_1 \vee l_2 \vee x_2), (\neg x_2 \vee l_3 \vee x_3), (\neg x_3 \vee l_4 \vee x_4) \dots (\neg x_{\ell-2} \vee l_{\ell-1} \vee l_\ell)$ using new variables $x_2, \dots, x_{\ell-2}$.
Let $G$ be the primal graph of the original formula, $G''$ be the primal graph of the new formula, and let $G'=G\cup G''$, i.e., a graph on the vertex set $V(G'')$ where the edge set is the union of the edge sets of $G$ and $G''$. By the construction it follows that $G'[V(G)]=G$ and that each connected component of $G'-V(G)$ has size at most $c-3$, where $c$ is the maximum clause size in the original formula. Moreover, the neighborhood of such a connected component is a clique in $G$. Since each clause induces a clique in the primal graph, it follows that $c\leq \pi(G)+1$ for each $\pi\in\{\td,\pw,\tw\}$. Hence, from Lemma~\ref{lem:incidenceBound} we have that $\pi(G')\leq \max(\pi(G),2\pi(G)-2)$, and an appropriate decomposition can be constructed from a decomposition of $G$ in logspace. Since $G''$ is a subgraph of $G$, it is also a decomposition of $G''$.

\medskip

\noindent $k$-SAT (primal graph) (for any fixed $k \in \N$) $\lred$ $k$-SAT (incidence graph):\\[0.1cm]
Use the same formula, bounds follow immediately from Lemma~\ref{lem:incidenceBound} (the connected components are single vertices).

\medskip

\noindent $k$-SAT (incidence graph) (for any fixed $k \in \N$) $\lred$ $k$-SAT (primal graph):\\[0.1cm]
The primal graph has a natural width-$k$ deconstruction into the incidence graph of the formula. For every variable of the formula we create a bag containing only it. For every clause of the formula we create a bag containing all the variables contained in this clause. It is easy to verify that this is a deconstruction.

\medskip

\noindent $k$-SAT (primal graph) (for any fixed $k \in \N$) $\lred$ CNF-SAT (primal graph):\\[0.1cm]
Trivial.

\medskip

\noindent 3-SAT (incidence graph) $\lred$ \textsc{3Coloring}:\\[0.1cm]
The reduction of Garey, Johnson, Stockmeyer~\cite{GareyJS76} creates a pair of adjacent vertices for every variable of the formula (a variable gadget), and a 6-vertex subgraph for every clause (clause gadget). For each clause, three edges are added to connect it to gadgets for variables occurring in this clause. Then a single triangle is created, whose one vertex is connected to all the vertices of all the variable gadgets. The graph created can easily be seen to have a width-11 deconstruction into the formula's incidence graph. Namely, a variable's (clause's) bag contains the 2 (6) corresponding vertices, and all bags contain the last triangle.

\medskip

\noindent \textsc{3Coloring} $\lred$ 3-SAT (primal graph):\\[0.1cm]
Create three variables $x,y,z$ and four clauses $(x \vee y \vee z),(\neg x \vee \neg y),(\neg y \vee \neg z),(\neg z \vee \neg x)$ for each vertex of the input graph $G$, describing that exactly one of the variables corresponding to this vertex is true. Then, for each edge of $G$ add three clauses of size $2$, describing that the true variable corresponding to one endpoint has a different label than for the other endpoint. It is easy to see that the formula's primal graph has a width-6 deconstruction into the original graph. Namely, for each original vertex of $G$ create a bag that contains the corresponding $3$ variables $x,y,z$.

\medskip

\noindent 3-SAT (incidence graph) $\lred$ \textsc{Independent Set}:\\[0.1cm]
Create two adjacent vertices $x,\neg x$ for every variable $x$ (variable gadget), and a triangle for every clause (clause gadget). In every clause gadget label the vertices of the triangle by the literals occurring in the clause, and connect these vertices to corresponding literals in clause gadgets. Then the input formula is satisfiable if and only if there is an independent set in the output graph with as many vertices as there are variables and clauses in total. The output graph has a trivial width-5 deconstruction into the incidence graph of the input formula, where for every variable/clause we create a bag containing the corresponding gadget.

\medskip

\noindent \textsc{Independent Set} $\lred$ \textsc{Vertex Cover}:\\[0.1cm]
Given a graph $G$ and a number $k$, output $G$ and $|V(G)|-k$.

\medskip

\noindent \textsc{Vertex Cover} $\lred$ \textsc{Dominating Set}:\\[0.1cm]
Given a graph $G$ and a number $k$, let $G'$ be obtained from $G$ by subdividing every edge once; output $G'$ and $k$. Bounds follow from Lemma~\ref{lem:incidenceBound} (the connected components are single vertices). 
\end{proof}

\end{document}